\documentclass{svjour3}

\usepackage{subfigure, charter, graphicx, amsmath, amssymb, color, enumerate, verbatim}
\usepackage{hyperref}

\newtheorem{theoremduplicate}{Theorem}

\newcommand{\stackmst}{\textsc{StackMST}}
\newcommand{\paths}[4]{ \mathcal{P}(#1, #2, #3, #4) }
\newcommand{\redpaths}[4]{ \widetilde{\mathcal{P}}(#1, #2, #3, #4) }
\newcommand{\bneck}[4]{ w(#1, #2, #3, #4) }
\newcommand{\redbneck}[4]{ \widetilde{w}(#1, #2, #3, #4) }
\newcommand{\intpaths}[4]{ \mathcal{P}_{int}(#1, #2, #3, #4) }
\newcommand{\intredpaths}[4]{ \widetilde{\mathcal{P}}_{int}(#1, #2, #3, #4) }
\newcommand{\intbneck}[4]{ w_{int}(#1, #2, #3, #4) }
\newcommand{\intredbneck}[4]{ \widetilde{w}_{int}(#1, #2, #3, #4) }
\newcommand{\etanode}[1]{ \tilde{#1} }
\DeclareMathOperator{\OPTtw}{OPT}
\DeclareMathOperator{\OPTsymbol}{OPT}
\newcommand{\OPT}[4]{\OPTsymbol_{#2}(#1,#3,#4)}

\newcommand{\pcomp}{parallel-compatible}
\newcommand{\scomp}{series-compatible}
\DeclareMathOperator{\mc}{mc}

\title{The Stackelberg Minimum Spanning Tree Game \\ on Planar and Bounded-Treewidth Graphs\thanks{A preliminary version of this paper appeared in~\cite{CDFJLNW-WINE}}}

\author{
Jean Cardinal
\and
Erik D. Demaine
\and
Samuel Fiorini
\and
Gwena\"el Joret\thanks{G. Joret is a Postdoctoral Researcher of the Fonds National de la Recherche Scientifique (F.R.S.--FNRS).}
\and
Ilan Newman
\and
Oren Weimann
}

\institute{J. Cardinal \at
Universit\'e Libre de Bruxelles (ULB), D\'epartement d'Informatique, CP~212\\
B-1050 Brussels, Belgium\\
\email{jcardin@ulb.ac.be}
\and
E. D. Demaine \at
MIT Computer Science and Artificial Intelligence Laboratory\\
Cambridge, MA 02139, USA\\
\email{edemaine@mit.edu}
\and
S. Fiorini \at
Universit\'e Libre de Bruxelles (ULB), D\'epartement de Math\'ematique, CP~216\\
B-1050 Brussels, Belgium\\
\email{sfiorini@ulb.ac.be}
\and
G. Joret \at
Universit\'e Libre de Bruxelles (ULB), D\'epartement d'Informatique, CP~212\\
B-1050 Brussels, Belgium\\
\email{gjoret@ulb.ac.be}
\and
I. Newman \at
Department of Computer Science, University of Haifa\\
Haifa 31905, Israel\\
\email{ilan@cs.haifa.ac.il}
\and
O. Weimann \at
Department of Computer Science, University of Haifa\\
Haifa 31905, Israel\\
\email{oren@cs.haifa.ac.il}
}

\begin{document}
\titlerunning{The Stackelberg Minimum Spanning Tree Game}
\sloppy

\date{}
\maketitle

\sloppy

\begin{abstract}
The Stackelberg Minimum Spanning Tree Game is a two-level combinatorial pricing problem played on a graph representing a network.
Its edges are colored either red or blue, and the red edges have a given fixed cost, representing the competitor's prices.
The first player chooses an assignment of prices to the blue edges, and the second player then buys the cheapest spanning tree,
using any combination of red and blue edges. The goal of the first player is to maximize the total price of purchased blue edges.

We study this problem in the cases of planar and bounded-treewidth graphs. We show that the problem is NP-hard on planar graphs
but can be solved in polynomial time on graphs of bounded treewidth.
\end{abstract}

\section{Introduction}

A young startup company has just acquired a collection of point-to-point tubes
between various sites on the Interweb.  The company's goal is to sell the use
of these tubes to a particularly stingy client, who will buy a minimum-cost
spanning tree of the network.  Unfortunately, the company has a direct
competitor: the government sells the use of a different collection of
point-to-point tubes at publicly known prices.
Our goal is to set the company's tube prices to maximize the company's income,
given the government's prices and the knowledge that the client will buy a
minimum spanning tree made from any combination of company and government tubes.
Naturally, if we set the prices too high, the client will rather buy the
government's tubes, while if we set the prices too low, we unnecessarily
reduce the company's income.

\begin{figure}
\begin{center}
\includegraphics[scale=.4]{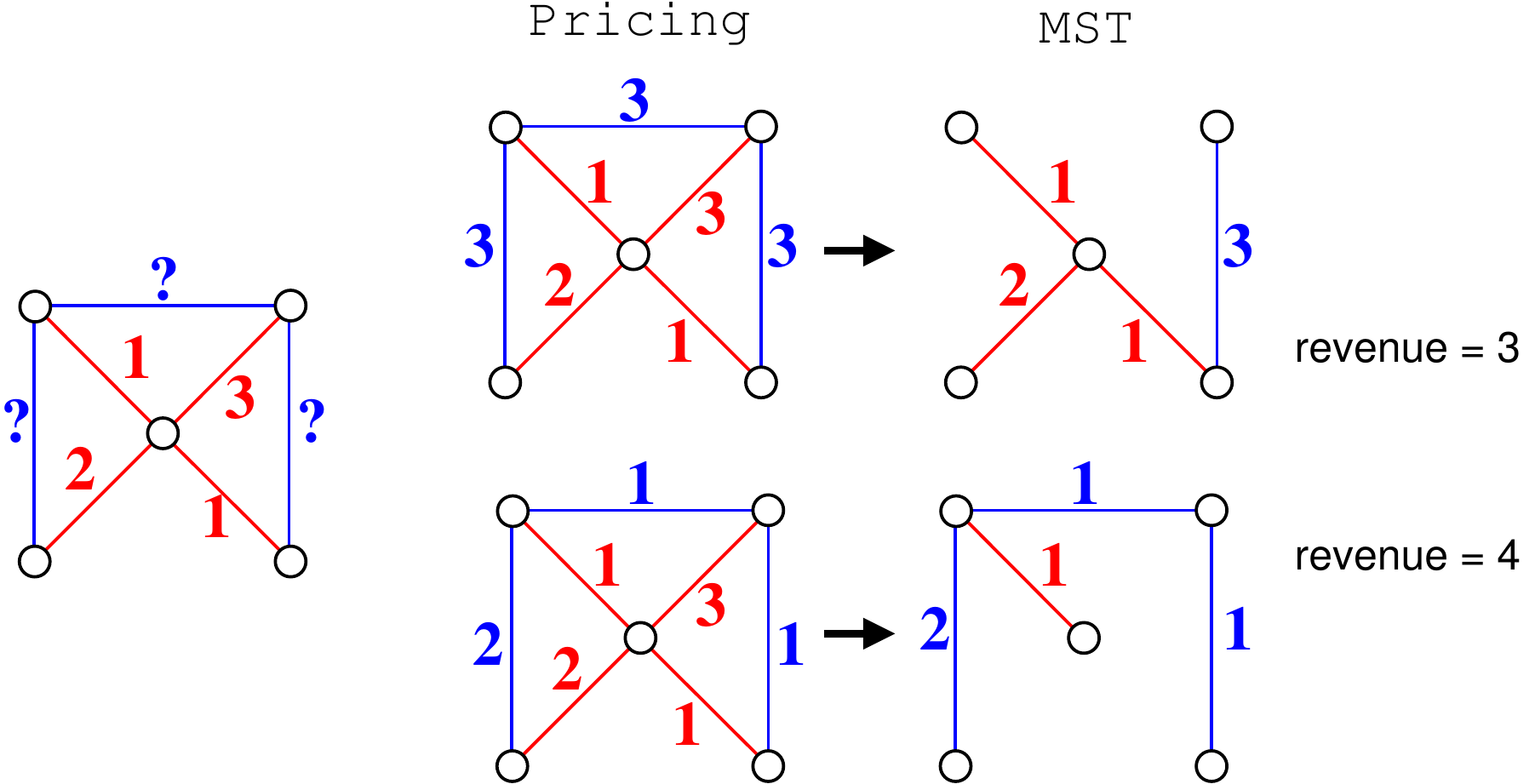}
\end{center}
\caption{\label{fig:example}A sample instance of the {\stackmst} problem.
  The goal is to assign prices to the blue edges to maximize the total
  price of the blue edges purchased in a minimum spanning tree.}
\end{figure}

This problem is called the {\em Stackelberg Minimum Spanning Tree Game}
\cite{CDFJLNW-sub}, and is an example in the growing family of algorithmic
game-theoretic problems about combinatorial optimization in graphs
\cite{GvHKUB04,RSM05,BHK08,vH06,BGPW08,GLSU09-journal,BHGV09,BGLP10,BCKLN10}.
More formally, we are given an undirected graph $G$
(possibly with parallel edges, but no loops), whose edge set $E(G)$ is
partitioned into a {\em red edge set} $R(G)$ and a {\em blue edge set} $B(G)$.
We are also given a cost function $c:R(G)\to {\mathbb R}^+$ assigning a
positive cost to each red edge.  The {\stackmst} problem is to assign
a price $p(e)$ to each blue edge~$e$,
resulting in a weighted graph $(G,c \cup p)$,
to maximize the total price of blue edges in a minimum spanning tree.
We assume that, if there is more than one minimum spanning tree,
we obtain the maximum possible income.
(Otherwise, we could decrease the prices slightly
and get arbitrarily close to the same income.)
Figure~\ref{fig:example} shows an example.

This problem is thus a two-player two-level optimization problem,
in which the leader (the company) chooses a strategy (a price assignment),
taking into account the strategy of the follower (the client),
which is determined by a second-level optimization problem
(the minimum spanning tree problem).
Such a game is known as a {\em Stackelberg game} in economics~\cite{Stack34}.

\paragraph{Motivations and scope.}

The Stackelberg Minimum Spanning Tree Game is a suitable model for real-life network pricing problems,
of the same flavor as those previously used for taxation and freight tariff-setting in the operations
research community (see for instance~\cite{LMS98,BLMS00,BLMS01}). It can be used to model
pricing in communication or transportation networks, and is easily amenable to meaningful generalizations
(see previous works below).

In this contribution, we aim at studying the problem under two natural restrictions. First, we
consider the class of planar instances, i.e., in which the input graph is planar. This can model
situations in which the input network corresponds to geographic connections.
Many important combinatorial optimization problems admit polynomial-time approximation schemes
on planar graphs. Among the first such results, Baker's technique~\cite{Baker94} is well known.
Since then, many more powerful techniques have been proposed~\cite{Klein05,Klein06,BKK07,DHM07,DHK09},
which ultimately rely on the ability to efficiently solve the problem in graphs of bounded treewidth in polynomial time.

This leads us to the second structural restriction we will tackle. Bounded-treewidth graphs have
the property of being ``close'' to trees, in the sense that they have can be augmented into chordal graphs with a bounded
clique number. They also constitute a natural structural restriction, that may be verified in real-life cases, and have proven fundamental
in many other combinatorial problems (see for instance the surveys from Bodlaender~\cite{B06} and Bodlaender and Koster~\cite{BK08}).

Optimization algorithms on bounded-treewidth graphs are generally based on dynamic programming,
using a textbook technique for well-behaved problems. In particular, it was shown by Courcelle~\cite{C08}
that the problem of checking a graph-theoretic property expressible in monadic second-order logic is fixed-parameter tractable
with respect to the treewidth of the graph. However, few if any such dynamic programs have been developed for
a bilevel optimization problem such as \stackmst, and standard techniques do not seem to apply. We expect our contribution
to give a basis for further application of graph decompositions to other bilevel optimization problems.

\paragraph{Previous results.}
The complexity and approximability of the \stackmst\ problem has been studied
in a previous paper~\cite{CDFJLNW-sub}. It was shown that
the problem is APX-hard, but can be approximated within a logarithmic factor.
Also, constant-factor approximation exist for the special cases in which the given
costs are bounded or take a bounded number of distinct values.  Finally, an
integer programming formulation has an integrality gap corresponding to the
best known approximation factors.

Briest et al. \cite{BHK08} generalized the above results to a wider class of pricing problems on graphs. This includes,
in particular, pricing problems with many followers and shortest path pricing games. They show that the single-price
strategy proposed in \cite{CDFJLNW-sub} yields logarithmic approximation factors for these games as well. They also tackle
a Stackelberg bipartite vertex cover game, which is shown to be solvable in polynomial time.

Recently, Bil\`o et al. \cite{BGLP10} studied special cases and another generalization of the \stackmst\ problem. In particular, they show that
the problem is approximable within a constant factor whenever the set of blue edges of $G$ forms a complete graph, and is solvable in polynomial time if,
additionally, there are only two distinct red costs. The generalization involves activation costs for the blue edges, and a leader
with a bounded activation budget. They generalize previous results to that case, and give an approximation factor parameterized by the
radius of the spanning tree induced by the red edges.

\paragraph{Our results.}
In Section~\ref{Planar Graphs}, we prove that {\stackmst}
remains NP-hard when restricted to planar graphs (Theorem~\ref{thm-hard}). The reduction is a strengthening from
our previous result, and is from the minimum connected vertex cover problem.

In Section~\ref{Series-Parallel Graphs}, we develop the tools required for the design of
a polynomial-time dynamic programming algorithm for {\stackmst} in series-parallel graphs.
These graphs have treewidth at most 2 and are planar, and they can be alternatively
defined in an inductive fashion using two composition operations. We show (Theorem~\ref{th-sp}) that the {\stackmst}
problem can be solved in $O(m^4)$ time  on series-parallel graphs with $m$ edges.

Finally, Section~\ref{Bounded-Treewidth Graphs} deals with graphs of arbitrary treewidth $t$. Our Theorem~\ref{th-tw} states
that the problem can be solved in $2^{O(t^3)}m + m^{O(t^2)}$ time on those graphs.

\section{Planar Graphs}
\label{Planar Graphs}

We consider the {\stackmst} problem on planar graphs. We strengthen the hardness result given in~\cite{CDFJLNW-sub} by showing that the problem
remains NP-hard in this special case. The reduction is from the {\em minimum connected vertex cover problem}, which is known to be NP-hard,
even when restricted to planar graphs of maximum degree 4 (see Garey and Johnson~\cite{GJ}). The minimum connected vertex cover problem consists
of finding a minimum-size subset $C$ of the vertices of a graph, such that every edge has at least one endpoint in $C$, and $C$ induces
a connected graph.

\begin{theorem}
\label{thm-hard}
The {\stackmst} problem is NP-hard, even when restricted to planar graphs.
\end{theorem}

\begin{proof}
Given a planar graph $G = (V,E)$, with $|V|=n$ and $|E|=m$, we construct an instance of {\stackmst} with red costs in $\{1,2\}$. Let $G'=(V',R\cup B)$ be the graph for this instance, with $(R,B)$ a bipartition of the edge set. We first let $V'=V\cup E$. The set of blue edges $B$ is the set $\{ ve : e\in E, v\in e  \}$. Thus the blue subgraph is the vertex-edge incidence graph of $G$, which is clearly planar. Given a planar embedding of the blue subgraph, we connect all vertices $e\in E$ of $G'$ by a tree, all edges of which are red and have cost $1$. The graph can be kept planar by letting those red edges be nonintersecting chords of the faces of the embedding. Finally, we double all blue edges by red edges of cost 2. The whole construction is illustrated in figure~\ref{fig:redu1}.

\begin{figure}
\begin{center}
\subfigure[\label{fig:redu1}The graphs $G$ and $G'$.]{\includegraphics[angle=-90, scale=.45]{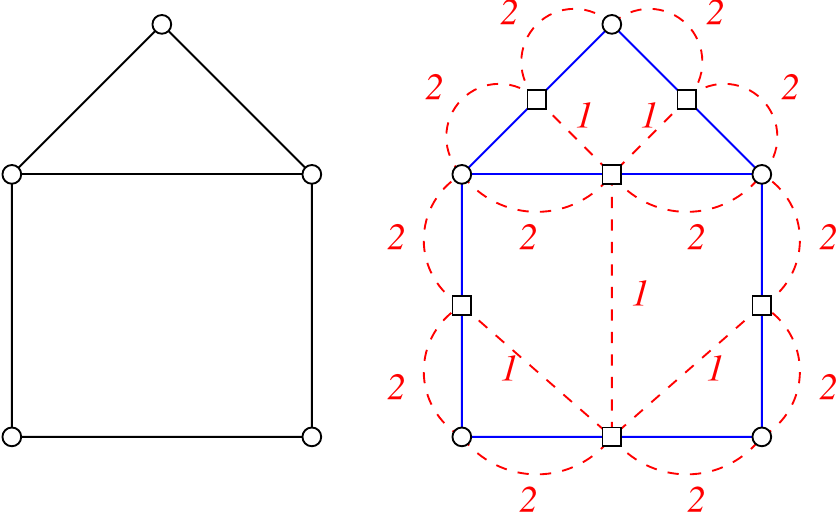}}
\hspace{1cm}
\subfigure[\label{fig:redu2}A connected vertex cover in $G$ and the corresponding price function in $G'$.]{\includegraphics[angle=-90, scale=.45]{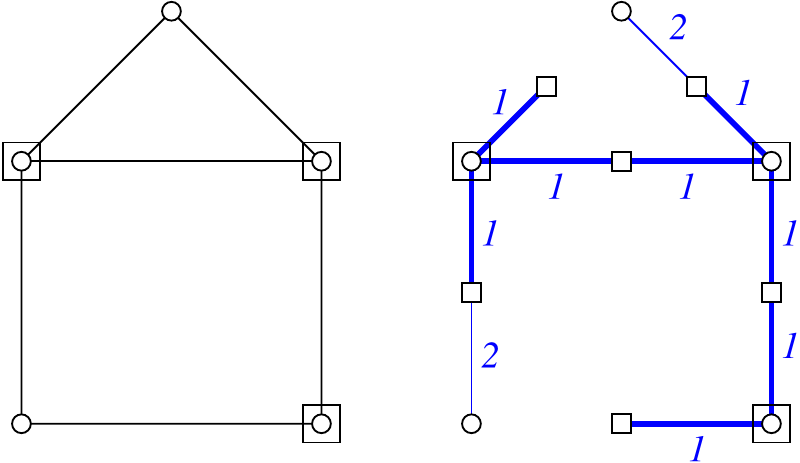}}
\end{center}
\caption{Illustration of the proof of Theorem~\ref{thm-hard}.}
\end{figure}


Let $t$ be a positive integer. We show that the revenue for an optimal price function for $G'$ is at least $m + 2n - t - 1$ if and only if there exists a connected vertex cover of $G$ of size at most $t$.

\medskip


$(\Leftarrow)$
We first suppose that there exists such a connected vertex cover $C\subseteq V$, and show how to construct a price function yielding the given revenue.

From the set $C$, we can construct a tree made of blue edges that spans all vertices $e\in E$ of $G'$. The set of vertices of this tree is $C\cup E$,
and its edges are of the form $ue\in E'$, with $u\in C$ and $e\in E$ (see figure~\ref{fig:redu2}).
This tree has $t+m-1$ blue edges, to which we assign price $1$. Now we have to connect the remaining $n-t$ vertices belonging to $V$. Since the only red edges incident to these vertices have cost $2$, we can use $n-t$ blue edges of price $2$ to include these vertices in the minimum spanning tree. The price of the other blue edges is set to $\infty$. The revenue for this price function is exactly $(t+m-1) + 2(n-t) = m+2n-t-1$.

\medskip


$(\Rightarrow)$
Now suppose that we have a price function yielding revenue at least $m + 2n - t - 1$. We can assume (see~\cite{CDFJLNW-sub}) that all the prices belong to the set $\{ 1,2,\infty \}$. We also assume that the price function is optimal and minimizes the number of red edges in the resulting spanning tree $T$.

First, we observe that $T$ does not contain any red edge.
By contradiction, if $T$ contains a red edge of cost 2, then this edge can be replaced by the parallel blue edge.
On the other hand, if $T$ contains a red edge $f$ of cost 1, we consider the cut defined by removing $f$ from $T$.
In the face used to define $f$, there exists a blue edge having its endpoints across the cut and does not belong to $T$.
So we can use this blue edge, with a price equal to 1, to reconnect the tree.

Now let us consider the blue edges of price $1$ in $T$.
We claim that the graph $H$ induced by these edges contains all vertices $e\in E$ of $G'$ and is connected.

Clearly, all vertices $e\in E$ of $G'$ are incident to a blue edge of price $1$, otherwise it can be reconnected to $T$ with a red edge of cost $1$, and $T$ is not minimum.
Thus $E \subseteq V(H)$, where $V(H)$ is the vertex set of $H$.
Letting $C := V(H)\cap V$, we conclude that $C$ is a vertex cover of the original graph $G$.

Now we show that $H$ is connected. Suppose otherwise; then there exist two vertices of $G'$ in $E$ that are connected by a red edge of cost $1$, and belonging to two different connected components $H_1$ and $H_2$ of $H$.
Consider the (blue) edge that connects $H_1$ and $H_2$ in $T$.
This edge cannot have price 2 in $T$, since $H_1$ and $H_2$ are connected by a red edge of cost $1$.
Hence the blue edge has price 1 and belongs to $H$.
Therefore $H$ is connected and $C$ is a connected vertex cover of $G$.

Finally the remaining vertices $V-C$ of $G'$ must be leaves of $T$, since otherwise they belong to a cycle containing a red edge of cost $1$. The total cost of $T$ is therefore $(m+|C|-1) + 2(n-|C|) = m + 2n - |C| - 1$. Since we know this is at least $m + 2n - t - 1$, we conclude that $|C|\leq t$.
\begin{flushright}\qed\end{flushright}\end{proof}

\section{Series-Parallel Graphs}
\label{Series-Parallel Graphs}

We now describe a polynomial-time dynamic programming algorithm for solving
the {\stackmst} problem on series-parallel graphs. These graphs are planar and have treewidth at most $2$.

We use the following inductive definition of (connected) series-parallel graphs. Consider a connected graph $G$ with two distinguished vertices
$s$ and $t$. The graph $(G,s,t)$ is a {\sl series-parallel} graph if
either $G$ is a single edge $(s,t)$, or $G$ is a {\sl series} or {\sl parallel} composition
of two series-parallel graphs $(G_{1},s_{1},t_{1})$ and $(G_{2},s_{2},t_{2})$.
The series composition of $G_1$ and $G_2$ is formed by setting  $s = s_{1}, t=t_{2}$ and identifying $t_1=s_2$; the parallel composition is formed by identifying $s = s_{1}=s_{2}$ and $t=t_{1}=t_{2}$.

We first give a number of useful lemmas and an outline of the dynamic programming algorithm. This algorithm will use two main rules, corresponding to the series and parallel composition operations. Once the two rules are defined, the description of the algorithm is straightforward.

\subsection{Preliminaries}
\label{sec-def}

Let us fix an instance of {\stackmst}, that is, a graph $G$ with $E(G) = R(G) \cup B(G)$ endowed
with a cost function $c: R(G) \to \mathbb{R}_{+}$. Denote by $c_{1}, c_{2}, \dots, c_{k}$ the different values taken
by $c$, in increasing order. Let also $c_{0}:=0$.\medskip

For two distinct vertices $s,t \in V(G)$ of $G$ and
a subset $F \subseteq B(G)$ of blue edges, define $\paths{G}{F}{s}{t}$
as the set of $st$-paths in the graph $(V(G), R(G) \cup F)$.
Let also $\redpaths{G}{F}{s}{t}$ denote the subset of paths in $\paths{G}{F}{s}{t}$
that contain at least one red edge.
A lemma of Cardinal {\it et al}.~\cite{CDFJLNW-sub} can be restated as follows.

\begin{lemma}[\cite{CDFJLNW-sub}]
\label{lem-opt-for-fixed-blue}
Suppose that $G$ contains a red spanning tree, and let $F\subseteq B(G)$ be
an acyclic subset of blue edges. Then, the maximum revenue achievable by the leader,
over solutions where the set of blue edges bought by the follower is exactly $F$,
is obtained by setting the price of each edge $st \not\in F$ to $+\infty$, and the price
of each edge $st\in F$ to
$$
\min\left\{ \max_{e\in P\cap R(G)} c(e) \mid P\in \redpaths{G}{F}{s}{t} \right\}.
$$
\end{lemma}

This lemma states that if we know the set of blue edges that will eventually be bought, the price of a selected
blue edge $st$ is given by the minimum, over the paths from $s$ to $t$, of the largest red cost on this path.

Motivated by this result, we introduce some more notations.
For a subset $Z \subseteq E(G)$ of edges, we define $\mc(Z)$ as the maximum
cost of a red edge in $Z$ if $Z \cap R(G) \neq \varnothing$, as $c_{0}=0$ otherwise.
(The two letters $\mc$ stand for ``max cost''.)
We define $\bneck{G}{F}{s}{t}$ as
$$
\bneck{G}{F}{s}{t}:= \left\{
\begin{array}{ll}
\min \left\{\mc(P) \mid P \in \paths{G}{F}{s}{t}\right\} &
\textrm{ if } \paths{G}{F}{s}{t} \neq \varnothing; \\
c_{k}  & \textrm{ otherwise}.
\end{array}
\right.
$$
Similarly,
$$
\redbneck{G}{F}{s}{t}:= \left\{
\begin{array}{ll}
\min \left\{\mc(P) \mid P \in \redpaths{G}{F}{s}{t}\right\} &
\textrm{ if } \redpaths{G}{F}{s}{t} \neq \varnothing; \\
c_{k}  & \textrm{ otherwise}.
\end{array}
\right.
$$
Thus, the price assigned to the edge $st\in F$ in Lemma~\ref{lem-opt-for-fixed-blue} is $\redbneck{G}{F}{s}{t}$.
Also, for the purpose of induction, we will consider graphs that do not necessarily contain a red spanning tree;
this is why we need to treat the case where $\paths{G}{F}{s}{t}$ or $\redpaths{G}{F}{s}{t}$ is empty in the above definitions.
\bigskip

In what follows, we let $[k] := \{ 0,1,\ldots , k\}$.
Our dynamic programming solution for series-parallel graphs associates a value to each pair $(H, q)$,
where $q\in [k]^2$, and $H$ is a graph appearing in the series-parallel decomposition of $G$.

A subset $F\subseteq B(G)$ of blue edges {\sl realizes} $q=(i,j)\in [k]^2$ in $(G,s,t)$ if
$F$ is acyclic and $\bneck{G}{F}{s}{t}=c_{i}$. Although this property does not depend on $j$,
the formulation will appear to be convenient. Similarly, we say that $q$ is {\sl realizable} in $(G,s,t)$ if
there exists such a subset $F$.

For $j\in [k]$ and distinct vertices $s,t\in V(G)$, let $G^+$ denote the graph
$G$ with an additional red edge between $s$ and $t$ of cost $c_j$.
We define
$$
\OPT{G}{(i,j)}{s}{t} := \max\left.\left\{ \sum_{uv \in F} \redbneck{G^{+}}{F}{u}{v}
\,\right| F\subseteq B(G), F \textrm{ realizes } (i,j)  \textrm{ in } (G,s,t) \right\},
$$
if such a subset $F$ exists, and set $\OPT{G}{(i,j)}{s}{t} := -\infty$ otherwise.\medskip

Intuitively, we want to keep track of optimal acyclic subsets of blue edges for every graph $G$
obtained during the construction of a series-parallel graph. The problem is, that the weights of the blue edges in the optimal solution
might change as we compose graphs in the series-parallel decomposition. However, the weights of edges
depend only on the maximum red costs, or {\em bottlenecks}, of the new $st$-paths that will be added to $G$. We can thus
prepare $\OPT{G}{}{s}{t}$ for every possible set of bottlenecks. These bottlenecks are the values $j$ in what precedes. The value
$i$ then corresponds to the new bottleneck that is realized, to be taken into account in future compositions.

Note that by Lemma~\ref{lem-opt-for-fixed-blue}, if $G$ has
a red spanning tree, then the maximum revenue achievable by the leader on instance $G$  equals
$$
\max_{i\in [k]}  \OPT{G}{(i,k)}{s}{t}.
$$
This will be the result returned by the dynamic programming solution.

\subsection{Series Compositions}
\label{sec-series}

Let $q=(i,j)$, $q_{1}=(i_{1},j_{1})$, and $q_{2}=(i_{2},j_{2})$, with $q,q_1,q_2\in [k]^2$.
We say that the pair $(q_{1},q_{2})$ is {\sl {\scomp}} with $q$ if
\begin{enumerate}[(S1)]
\item \label{cond-series-i} $\max\{i_{1}, i_{2}\} = i$;
\item \label{cond-series-j1} $\max\{j, i_{2}\} = j_{1}$, and
\item $\max\{j, i_{1}\} = j_{2}$,
\end{enumerate}
Notice that $(q_{1},q_{2})$ is {\scomp} with $q$ if and only if
$(q_{2},q_{1})$ is.

This condition allows us to use the following recursion in our dynamic programming algorithm.
\begin{lemma}
\label{lem-opt-series}
Suppose that
$(G,s,t)$ is a series composition of $(G_{1},s_{1},t_{1})$ and $(G_{2},s_{2},t_{2})$, and that
$q \in [k]^2$ is realizable in $(G,s,t)$. Then
$$
\OPT{G}{q}{s}{t} = \max \left\{ \OPT{G_{1}}{q_{1}}{s_{1}}{t_{1}} + \OPT{G_{2}}{q_{2}}{s_{2}}{t_{2}} \mid (q_1,q_2)\text{\ is\ \scomp\ with\ }q \right\}.
$$
\end{lemma}

We now prove that the recursion is valid. We need the following lemmas. In what follows,
$(G,s,t)$ is a series composition of $(G_{1},s_{1},t_{1})$ and $(G_{2},s_{2},t_{2})$;
$q, q_{1}, q_{2} \in [k]^2$ with $q=(i,j)$, $q_{1}=(i_{1},j_{1})$, and $q_{2}=(i_{2},j_{2})$
are such that $(q_{1},q_{2})$ is {\scomp} with $q$;
and $F_{\ell} \subseteq B(G_{\ell})$ realizes $q_{\ell}$ in $(G_{\ell},s,t)$, for $\ell=1,2$.

We first observe that $F := F_{1} \cup F_{2}$ realizes $q$.

\begin{lemma}
\label{lem-series-F}
$F$ realizes $q$  in $(G,s,t)$.
\end{lemma}
\begin{proof}
Since $V(G_{1}) \cap V(G_{2}) = \{t_{1}\}$ ($=\{s_{2}\}$),
the set $F$ is clearly acyclic. It remains to show $\bneck{G}{F}{s}{t}=c_{i}$.
Every $st$-path in $\paths{G}{F}{s}{t}$ is the combination of
an $s_{1}t_{1}$-path  of $\paths{G_{1}}{F_{1}}{s_{1}}{t_{1}}$ with
an $s_{2}t_{2}$-path  of $\paths{G_{2}}{F_{2}}{s_{2}}{t_{2}}$.
It follows
$$
\bneck{G}{F}{s}{t} = \max \left\{ \bneck{G_{1}}{F_{1}}{s_{1}}{t_{1}},
\bneck{G_{2}}{F_{2}}{s_{2}}{t_{2}}\right\}
= \max\{c_{i_{1}},c_{i_{2}}\}
= c_{i},
$$
where the last equality is from~(S\ref{cond-series-i}).
\begin{flushright}\qed\end{flushright}\end{proof}

The proof of the next lemma is illustrated on Figure~\ref{fig-series}. It motivates the definition of series-compatibility.
\begin{lemma}
\label{lem-series-bneck}
Let $G^{+}$ be the graph $G$ augmented with a red edge $st$ of cost $c_j$,
and $G_{\ell}^{+}$ (for $\ell = 1,2$) the graph $G_{\ell}$ augmented with a red edge $s_{\ell}t_{\ell}$ of cost $c_{j_{\ell}}$.
Then for $\ell=1,2$ and every edge $uv \in F_{\ell}$,
$$
\redbneck{G^{+}}{F}{u}{v} = \redbneck{G_{\ell}^{+}}{F_{\ell}}{u}{v}.
$$
\end{lemma}
\begin{proof}
We prove the statement for $\ell = 1$, the case $\ell = 2$ follows by symmetry.
Let $uv \in F_{1}$, and let $e=st$ and $e_{1}=s_{1}t_{1}$ be the
additional red edges in  $G^{+}$ and $G^{+}_{1}$, respectively.

We first show:

\begin{claim}
\label{claim-series-bneck-geq}
$\redbneck{G^{+}}{F}{u}{v} \geq \redbneck{G_{1}^{+}}{F_{1}}{u}{v}$.
\end{claim}
\begin{proof}
The claim is true if $\redpaths{G^{+}}{F}{u}{v}=\varnothing$, since then
$\redbneck{G^{+}}{F}{u}{v} = c_{k} \geq \redbneck{G_{1}^{+}}{F_{1}}{u}{v}$.
Suppose thus $\redpaths{G^{+}}{F}{u}{v} \neq \varnothing$,
and let $P \in \redpaths{G^{+}}{F}{u}{v}$.
It is enough to show that $\mc(P) \geq \redbneck{G_{1}^{+}}{F_{1}}{u}{v}$.
This clearly holds if $e \notin E(P)$, as $P$ belongs then
also to $\redpaths{G_{1}^{+}}{F_{1}}{u}{v}$ (recall that $|V(G_{1}) \cap V(G_{2})|=1$).
Hence, we may assume $e\in E(P)$. It follows $s_{1},t_{1} \in V(P)$.

Let $s_{1}Pt_{1}$ denote the subpath of $P$ comprised between $s_1$ and $t_1$.
Also let $P_{1}$ denote the path of $\redpaths{G_{1}^{+}}{F_{1}}{u}{v}$ obtained by replacing
the subpath $s_{1}Pt_{1}$ of $P$ with the edge $e_{1}$.
Using~(S\ref{cond-series-j1}), we obtain
$$
\mc(s_{1}Pt_{1}) = \max\{c_{j}, \mc(t_{2}Pt_{1})\} \geq \max\{c_{j}, c_{i_{2}}\}=c_{j_{1}},
$$
implying $\mc(P) \geq \mc(P_{1}) \geq \redbneck{G_{1}^{+}}{F_{1}}{u}{v}$.
\begin{flushright}\qed\end{flushright}\end{proof}

Conversely, we prove:

\begin{claim}
\label{claim-series-bneck-leq}
$\redbneck{G^{+}}{F}{u}{v} \leq \redbneck{G_{1}^{+}}{F_{1}}{u}{v}$.
\end{claim}
\begin{proof}
Again, this trivially holds if $\redpaths{G_{1}^{+}}{F_{1}}{u}{v}$ is empty.
Suppose thus $\redpaths{G_{1}^{+}}{F_{1}}{u}{v} \neq \varnothing$, and
let $P_{1} \in \redpaths{G_{1}^{+}}{F_{1}}{u}{v}$. Similarly as before, it is enough to show
that $\redbneck{G^{+}}{F}{u}{v} \leq \mc(P_{1})$. This is true
if $e_{1} \notin E(P_{1})$, since then $P_{1} \in \redpaths{G^{+}}{F}{u}{v}$.
Assume thus  $e_{1} \in E(P_{1})$.

If $\paths{G_{2}}{F_{2}}{s_{2}}{t_{2}} = \varnothing$, then $i_{2}=k$ and
$\mc(P_{1}) \geq c_{j_{1}} = \max \{c_{j}, c_{i_{2}}\} = c_{k} \geq \redbneck{G^{+}}{F}{u}{v}$
by~(S\ref{cond-series-j1}).
We may thus assume that $\paths{G_{2}}{F_{2}}{s_{2}}{t_{2}}$   contains a
path $P_{2}$; we choose $P_{2}$ such that $\mc(P_{2}) = c_{i_{2}}$.

Denote by $P$ the path obtained
from $P_{1}$ by replacing the edge $e_{1}$ with the combination
of edge $e$ and path $P_{2}$. Since $P \in \redpaths{G^{+}}{F}{u}{v}$,
(S\ref{cond-series-j1}) yields
\begin{align*}
\mc(P_{1})  &= \max\left\{ c_{j_{1}}, \mc(P_{1} - e_{1})\right\} \\
&= \max\left\{ c_{j}, c_{i_{2}}, \mc(P_{1} - e_{1})\right\} \\
&= \max\left\{c_{j}, \mc(P_{2}), \mc(P_{1} - e_{1})\right\} \\
&= \mc(P) \\
&\geq \redbneck{G^{+}}{F}{u}{v}.
\end{align*}
\begin{flushright}\qed\end{flushright}\end{proof}

The lemma follows from Claims~\ref{claim-series-bneck-geq} and~\ref{claim-series-bneck-leq}.
\begin{flushright}\qed\end{flushright}\end{proof}

\begin{figure}
\begin{center}
\includegraphics[scale=.3, angle=-90]{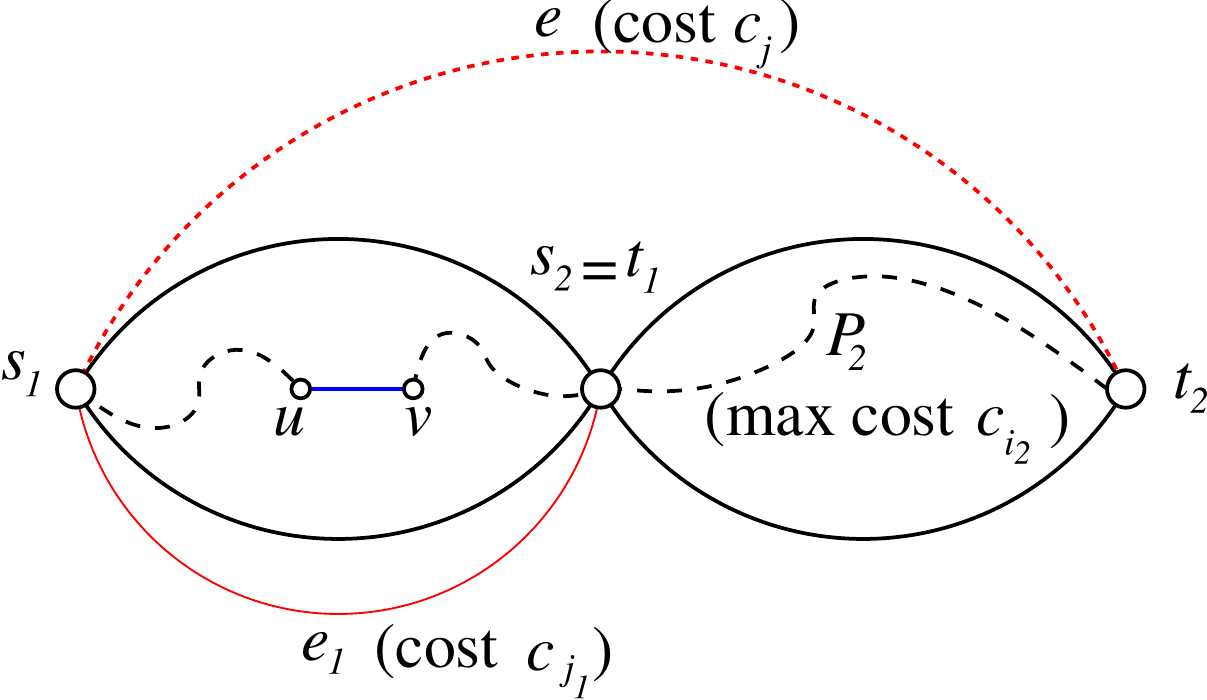}
\end{center}
\caption{\label{fig-series}Series composition: illustration of the proof of Lemma~\ref{lem-series-bneck}.}
\end{figure}

We are now ready to prove the correctness of the recursion step in Lemma~\ref{lem-opt-series}.
\begin{proof}[Proof of Lemma~\ref{lem-opt-series}]
Let $q$ and $G^{+}$ be defined as before.
We first show:
\begin{claim}
\label{claim-opt-series-exists}
There exist $q_{1}, q_{2}\in [k]^2$ such that $(q_{1},q_{2})$ is {\scomp} with $q$
and $\OPT{G}{q}{s}{t} \leq \OPT{G_{1}}{q_{1}}{s}{t} + \OPT{G_{2}}{q_{2}}{s}{t}$.
\end{claim}
\begin{proof}
Let $F\subseteq B(G)$ be a subset of blue edges realizing
$q$ in $(G,s,t)$ such that
$$
\OPT{G}{q}{s}{t} = \sum_{uv \in F} \redbneck{G^{+}}{F}{u}{v}.
$$
For $\ell=1,2$, let also $F_{\ell}:=F \cap E(G_{\ell})$ and $q_{\ell} :=(i_{\ell}, j_{\ell})$, with
$i_{\ell}$ the index such that $c_{i_{\ell}} = \bneck{G_{\ell}}{F_{\ell}}{s_{\ell}}{t_{\ell}}$, and $j_{\ell} := \max\{j, i_{\ell+1}\}$
(indices are taken modulo 2).
$F_{\ell}$ ($\ell=1,2$) clearly realizes $q_{\ell}$  in $(G_{\ell},s_{\ell},t_{\ell})$.
It is also easily verified that $(q_{1}, q_{2})$
is {\scomp} with $q$. Hence we can apply Lemma~\ref{lem-series-bneck}:
\begin{align*}
\OPT{G}{q}{s}{t} &= \sum_{uv \in F} \redbneck{G^{+}}{F}{u}{v} \\
 &= \sum_{uv \in F_{1}} \redbneck{G^{+}_{1}}{F_{1}}{u}{v}
+  \sum_{uv \in F_{2}} \redbneck{G^{+}_{2}}{F_{2}}{u}{v} \\
&\leq \OPT{G_{1}}{q_{1}}{s_{1}}{t_{1}} + \OPT{G_{2}}{q_{2}}{s_{2}}{t_{2}},
\end{align*}
as claimed.
\begin{flushright}\qed\end{flushright}\end{proof}

We now prove:

\begin{claim}
\label{claim-opt-series-forall}
$\OPT{G}{q}{s}{t} \geq \OPT{G_{1}}{q_{1}}{s_{1}}{t_{1}} + \OPT{G_{2}}{q_{2}}{s_{2}}{t_{2}}$
holds for every $q_{1}, q_{2}\in [k]^2$ such that $(q_{1},q_{2})$ is {\scomp} with $q$.
\end{claim}
\begin{proof}
Suppose that $(q_{1},q_{2})$ is {\scomp} with $q$.
Let $F_{\ell}\subseteq B(G_{\ell})$
($\ell=1,2$) be a subset of blue edges of $G_{\ell}$ such that
$$
\OPT{G_{\ell}}{q_{\ell}}{s_{\ell}}{t_{\ell}}
= \sum_{uv \in F_{\ell}} \redbneck{G^{+}_{\ell}}{F_{\ell}}{u}{v}.
$$
By Lemma~\ref{lem-series-F}, $F:= F_{1} \cup F_{2}$ realizes
$q$ in $(G,s,t)$. Using again Lemma~\ref{lem-series-bneck}, we have:
\begin{align*}
\OPT{G}{q}{s}{t} &\geq \sum_{uv \in F} \redbneck{G^{+}}{F}{u}{v} \\
&= \sum_{uv \in F_{1}} \redbneck{G^{+}_{1}}{F_{1}}{u}{v}
+  \sum_{uv \in F_{2}} \redbneck{G^{+}_{2}}{F_{2}}{u}{v} \\
&= \OPT{G_{1}}{q_{1}}{s_{1}}{t_{1}} + \OPT{G_{2}}{q_{2}}{s_{2}}{t_{2}},
\end{align*}
and the claim follows.
\begin{flushright}\qed\end{flushright}\end{proof}

The lemma follows from Claims~\ref{claim-opt-series-exists} and~\ref{claim-opt-series-forall}.
\begin{flushright}\qed\end{flushright}\end{proof}

\subsection{Parallel Compositions}
\label{sec-parallel}

The recursion step for parallel compositions follows a similar scheme.
Let $q,q_1,q_2\in [k]^2$ with $q=(i,j)$, $q_{1}=(i_{1},j_{1})$, and $q_{2}=(i_{2},j_{2})$.
We say that the pair $(q_{1},q_{2})$ is {\sl {\pcomp}} with $q$ if
\begin{enumerate}[(P1)]
\item \label{cond-parallel-one-false} at least one of $i_{1},i_{2}$ is non-zero;
\item \label{cond-parallel-i} $\min\{i_{1}, i_{2}\} = i$;
\item \label{cond-parallel-j1} $\min\{j, i_{2}\} = j_{1}$, and
\item $\min\{j, i_{1}\} = j_{2}$,
\end{enumerate}
The recursion step for parallel composition is as follows.

\begin{lemma}
\label{lem-opt-parallel}
Suppose that $(G,s,t)$ is a parallel composition of $(G_{1},s,t)$ and $(G_{2},s,t)$, and that
$q \in [k]^2$ is realizable in $(G,s,t)$. Then
$$
\OPT{G}{q}{s}{t} = \max \{ \OPT{G_{1}}{q_{1}}{s}{t} + \OPT{G_{2}}{q_{2}}{s}{t} \mid
(q_{1},q_{2})\text{\ is \pcomp\ with\ } q \}.
$$
\end{lemma}

In what follows, $(G,s,t)$ is a parallel composition of $(G_{1},s,t)$ and $(G_{2},s,t)$;
$(q_{1},q_{2})$ is {\pcomp} with $q$; and
$F_{\ell} \subseteq B(G_{\ell})$ realizes $q_{\ell}$ in $(G_{\ell},s,t)$, for $\ell=1,2$.
Also, $F := F_{1} \cup F_{2}$.

Similarly to Lemma~\ref{lem-series-F}, the definition of parallel-compatibility implies the following lemma.
\begin{lemma}
\label{lem-parallel-F}
$F$ realizes $q$ in $(G,s,t)$.
\end{lemma}
\begin{proof}
We have to prove that $F$ is acyclic and that $\bneck{G}{F}{s}{t}=c_{i}$.

First, suppose that $(V(G), F)$ contains a cycle $C$.
Since $F_1$ and $F_{2}$ are both acyclic, $C$ includes the vertices $s$ and $t$, and moreover
$E(G_1) \cap E(C)$, $E(G_2) \cap E(C)$ are both non-empty.
But then, there is an $st$-path in $(V(G), F_{\ell})$ for $\ell=1,2$,
implying $i_{1}=i_{2}=0$, which contradicts
(P\ref{cond-parallel-one-false}). Hence, $F$ is acyclic.

Now, since each path of $\paths{G}{F}{s}{t}$ is included in
either $\paths{G_{1}}{F_{1}}{s}{t}$ or $\paths{G_{2}}{F_{2}}{s}{t}$,
it follows $\bneck{G}{F}{s}{t} =
\min\{ \bneck{G_{1}}{F_{1}}{s}{t}, \bneck{G_{2}}{F_{2}}{s}{t}\} = \min\{c_{i_{1}},c_{i_{2}}\}$,
which equals $c_{i}$ by (P\ref{cond-parallel-i}).
\begin{flushright}\qed\end{flushright}\end{proof}

The next lemma is the analogue of Lemma~\ref{lem-series-bneck} for parallel compositions.
\begin{lemma}
\label{lem-parallel-bneck}
Let $G^{+}$ be the graph $G$ augmented with a red edge $st$ of cost $c_j$,
and let $G_{\ell}^{+}$ (for $\ell = 1,2$) be the graph $G_{\ell}$ augmented with a red edge $s_{\ell}t_{\ell}$ of cost $c_{j_{\ell}}$.
Then for $\ell=1,2$ and every edge $uv \in F_{\ell}$,
$$
\redbneck{G^{+}}{F}{u}{v} = \redbneck{G_{\ell}^{+}}{F_{\ell}}{u}{v}.
$$
\end{lemma}
\begin{proof}
We prove the statement for $\ell=1$, the case $\ell=2$ follows by symmetry.
Let $e=st$ and $e_{1}=s_{1}t_{1}$ be the additional red edges in $G^{+}$ and $G^{+}_{1}$, respectively.

Let $uv \in F_{1}$.
Observe  that $\redpaths{G^{+}}{F}{u}{v}$ is empty if and only if
$\redpaths{G_{1}^{+}}{F_{1}}{u}{v}$ is. If both are empty, then
$\redbneck{G^{+}}{F}{u}{v} = \redbneck{G_{1}^{+}}{F_{1}}{u}{v}=c_{k}$, and the claim holds.
Hence, we may assume
$\redpaths{G^{+}}{F}{u}{v} \neq \varnothing$ and  $\redpaths{G_{1}^{+}}{F_{1}}{u}{v} \neq \varnothing$.

We first show:

\begin{claim}
\label{claim-par-leq}
$\redbneck{G^{+}}{F}{u}{v} \leq \redbneck{G_{1}^{+}}{F_{1}}{u}{v}$.
\end{claim}
\begin{proof}
Let $P_{1} \in \redpaths{G_{1}^{+}}{F_{1}}{u}{v}$.
It is enough to show $\redbneck{G^{+}}{F}{u}{v} \leq \mc(P_{1})$.
If $e_{1} \notin E(P_{1})$, then $P_{1} \in \redpaths{G^{+}}{F}{u}{v}$, and
$\redbneck{G^{+}}{F}{u}{v} \leq \mc(P_{1})$ holds by definition. Hence we may assume
$e_{1} \in E(P_{1})$.

By (P\ref{cond-parallel-j1}), we have $j_{1}=\min\{j,i_{2}\}$. If $j_{1}=j$, then replacing
the edge $e_{1}$ of $P_{1}$ by $e$ yields
a path $P \in \redpaths{G^{+}}{F}{u}{v}$ with $\mc(P) = \mc(P_{1})$,
implying $\redbneck{G^{+}}{F}{u}{v} \leq \mc(P_{1})$. Similarly, if $j_{1}=i_{2} < j$, then
$i_{2} < k$, implying that $\paths{G_{2}}{F_{2}}{s}{t}$ is not empty.
Replacing in $P_{1}$ the edge $e_1$ with any path
$P_{2} \in \paths{G_{2}}{F_{2}}{s}{t}$
with  $\mc(P_{2})=c_{i_{2}}$ gives again a
path $P$ with $\mc(P) = \mc(P_{1})$.
While the path $P_{2}$ does not necessarily contain a red edge,
the path $P$, on the other hand, cannot be completely blue. This is
because otherwise $F$ contains the
cycle $P \cup \{uv\}$, contradicting the fact that $F$ is acyclic (as follows
from Lemma~\ref{lem-parallel-F}). Hence,
$P\in \redpaths{G^{+}}{F}{u}{v}$, and
$\redbneck{G^{+}}{F}{u}{v} \leq \mc(P) = \mc(P_{1})$.
Claim~\ref{claim-par-leq} follows.
\begin{flushright}\qed\end{flushright}\end{proof}

Conversely, we prove:

\begin{claim}
\label{claim-par-geq}
$\redbneck{G^{+}}{F}{u}{v} \geq \redbneck{G_{1}^{+}}{F_{1}}{u}{v}$.
\end{claim}
\begin{proof}
Let $P \in \redpaths{G^{+}}{F}{u}{v}$. Again, it is enough to
show $\mc(P) \geq \redbneck{G_{1}^{+}}{F_{1}}{u}{v}$.
This clearly holds if $P \in \redpaths{G^{+}_{1}}{F_{1}}{u}{v}$.
Hence, we may assume $s,t \in V(P)$, and that the subpath $sPt$ of $P$ either
belongs to $\paths{G_{2}}{F_{2}}{s}{t}$, or corresponds to the edge $e$ (by $sPt$ we denote the subpath of $P$ that is between vertices $s$ and $t$).

In the first case, $c_{i_{2}} \leq \mc(sPt)$ holds by definition.
Moreover, $j_{1} \leq i_{2}$ follows from (P\ref{cond-parallel-j1}). Therefore, replacing
the subpath $sPt$ of $P$ with the edge $e_{1}$ yields
a path $P_{1}\in \redpaths{G_{1}^{+}}{F_{1}}{u}{v}$
with $\mc(P_{1}) \leq \mc(P)$, implying $\redbneck{G_{1}^{+}}{F_{1}}{u}{v} \leq \mc(P)$.

Similarly, (P\ref{cond-parallel-j1}) implies $j_{1} \leq j$ in the second case.
Hence, replacing the edge $e$ of $P$ with $e_{1}$ results in
a path $P_{1}\in \redpaths{G_{1}^{+}}{F_{1}}{u}{v}$ with
$\mc(P_{1}) \leq \mc(P)$, showing $\redbneck{G_{1}^{+}}{F_{1}}{u}{v} \leq \mc(P)$. This completes
the proof of Claim~\ref{claim-par-geq}.
\begin{flushright}\qed\end{flushright}\end{proof}

Lemma~\ref{lem-parallel-bneck} follows from Claims~\ref{claim-par-leq} and~\ref{claim-par-geq}.
\begin{flushright}\qed\end{flushright}\end{proof}

Using the two previous lemmas, the proof of Lemma~\ref{lem-opt-parallel} is
the same as that of Lemma~\ref{lem-opt-series} for series composition. We omit it.

\subsection{The Algorithm}
\label{sec-algo}

\begin{theorem}
\label{th-sp}
The {\stackmst} problem can be solved in $O(m^4)$ time on series-parallel graphs.
\end{theorem}
\begin{proof}
A series-parallel decomposition of a connected
series-parallel graph can be computed in linear time~\cite{VTL82}.
Given such a decomposition, Lemmas~\ref{lem-opt-series} and~\ref{lem-opt-parallel}
yield the following algorithm: consider each graph $(H,s,t)$ in the decomposition
tree in a bottom-up fashion.

If $H$ is a single edge $st$, we directly compute $\OPT{H}{q}{s}{t}$ for every $q\in [k]^2$.
In particular, if $H$ is a single red edge of cost $c_h$, then $\OPT{H}{(i,j)}{s}{t} = 0$ if $i=h$, and $-\infty$ otherwise. On the other hand, if $H$ is a single blue edge, then $\OPT{H}{(i,j)}{s}{t}$ is equal to $c_j$ if $i=0$ (corresponding to the case $F=\{st\}$), to $0$ if $i=k$ (corresponding to the case $F=\varnothing$), and to $-\infty$ otherwise.

If $(H,s,t)$ is a series or parallel composition of
$(H_{1},s_{1},t_{1})$ and $(H_{2},s_{2},t_{2})$,
compute $\OPT{H}{q}{s}{t}$ for every $q\in [k]^2$ based on the
previously computed values for $(H_{1},s_{1},t_{1})$ and $(H_{2},s_{2},t_{2})$,
relying on Lemmas~\ref{lem-opt-series} and~\ref{lem-opt-parallel}.

For every $q=(i,j)\in[k]^2$, there are $O(k)$ possible values for either series-compatible
or parallel-compatible pairs $(q_1,q_2)$. Hence every step costs $O(k)$ times. Since
there are $O(k^2)$ possible values for $q$, and $O(m)$ graphs in the decomposition of $G$,
the overall complexity is $O(k^3m) = O(m^4)$.

This results in a polynomial-time algorithm computing the maximum revenue
achievable by the leader. Moreover, using Lemmas~\ref{lem-series-F} and~\ref{lem-parallel-F},
it is not difficult to keep track at each step of a witness $F \subseteq B(H)$ for
$\OPT{H}{q}{s}{t}$, whenever $\OPT{H}{q}{s}{t} > -\infty$.
This proves the theorem.
\begin{flushright}\qed\end{flushright}\end{proof}

\begin{figure}
\begin{center}
\includegraphics[scale=.68]{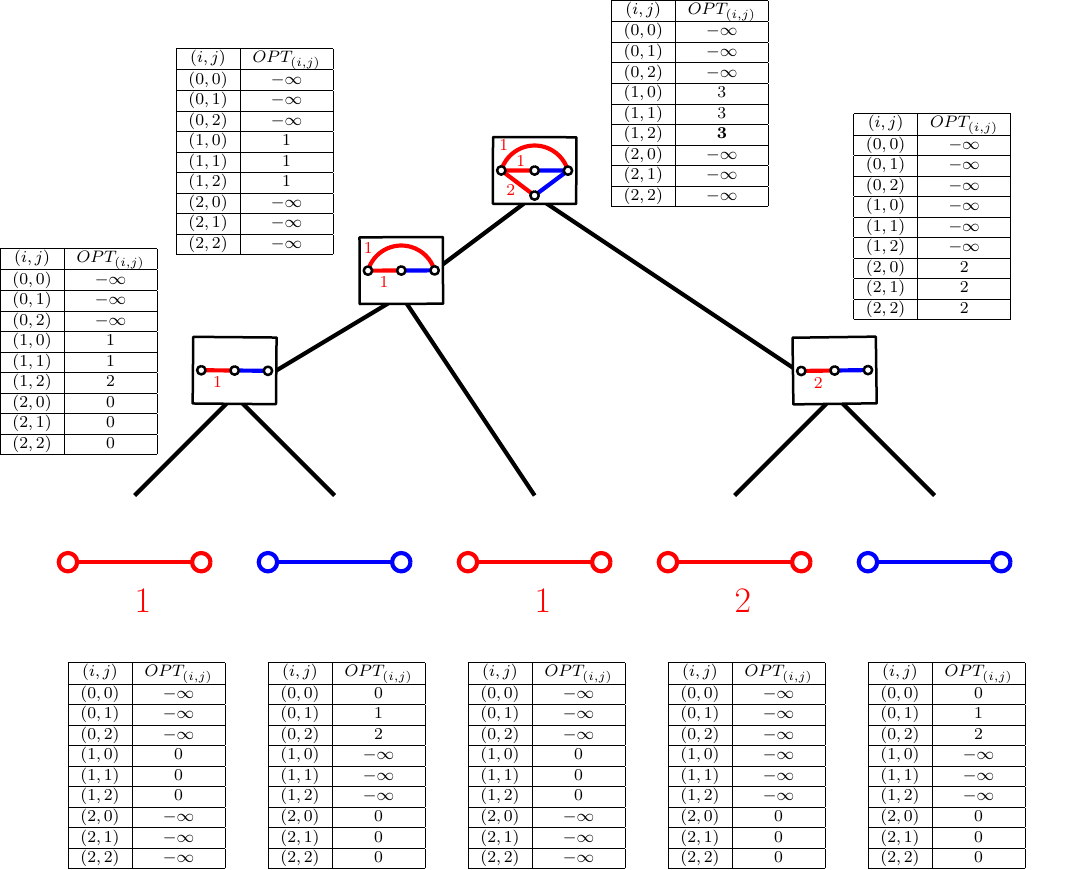}
\end{center}
\caption{\label{fig-ex}An example of execution of the dynamic programming algorithm for \stackmst\ on series-parallel graphs. The graph is constructed using two series compositions and two parallel compositions. The pairs $(i,j),OPT_{(i,j)}(H,s,t)$ are shown for each intermediate graph $(H,s,t)$ of the decomposition. The value 3 shown in boldface in the top table is the maximum achievable profit.}
\end{figure}

An example of execution of the algorithm is given in figure~\ref{fig-ex}.

\section{Bounded-Treewidth Graphs}
\label{Bounded-Treewidth Graphs}

In the previous section, we gave a polynomial-time algorithm for solving the
{\stackmst} problem on series-parallel graphs, which have
treewidth at most $2$.  In this section, we extend the algorithm to handle
graphs of bounded treewidth, as indicated by the following theorem.

\begin{theorem}
\label{th-tw}
The {\stackmst} problem can be solved in $2^{O(t^3)}m + m^{O(t^2)}$
time on graphs of treewidth $t$.
\end{theorem}

The treewidth of a graph $G$ is usually defined as the minimum width
of a tree decomposition of $G$. Since we will not use tree decompositions explicitly,
we skip the definition  (see for instance~\cite{D05}).
Instead we will rely on the fact, first proved by
Abrahamson and Fellows~\cite{AbrahamsonFellows1995},
that every graph of treewidth $t$ is isomorphic to a {\sl $t$-boundaried graph}, which is defined
as a graph with $t$
distinguished vertices (called {\sl boundary vertices}), each uniquely labeled by a label in
$\{1,\dots,t\}$, which can be build recursively using the following  operators:
\begin{enumerate}
\item The null operator $\varnothing$ creates an  $t$-boundaried graph
having only $t$ boundary vertices, and they are all isolated.
\item The binary operator $\oplus$ takes the disjoint union
of two $t$-boundaried graphs and identify the $i$th
boundary vertex of the first graph with the $i$th boundary vertex of
the second graph.
Thus the edges between two boundary vertices of $G_{1} \oplus G_{2}$ correspond to
the union of the edges between these vertices in $G_{1}$ and in $G_{2}$.
(Observe that this operation is exactly a parallel-composition
if there are only two boundary vertices.)
\item The unary operator $\eta$ introduces a new isolated vertex and makes
this the new vertex with label 1 in the boundary. The previous
vertex that was labeled 1 is removed from the boundary (but not from
the graph).
\item The unary operator $\epsilon$ adds an edge between the vertices labeled
1 and 2 in the boundary.
\item Unary operators that permute the labels of the boundary vertices.
\end{enumerate}

We note that, conversely, every $t$-boundaried graph has treewidth at most $t$
(but not necessarily exactly $t$).
The set of boundary vertices of a $t$-boundaried graph $G$ is denoted by $\partial(G)$.
Every $t$-boundaried graph
on $n$ vertices can be constructed by applying
$O(t n)$ compositions according to the above five operators.
This construction as well as the boundary vertices can be found in
$2^{O(t^3)}m$ time~\cite{Bodlaender1996} (note that is linear time is $t$ is a fixed constant).

To summarize, in order to prove Theorem~\ref{th-tw}, it is enough to show that
the {\stackmst} problem can be solved in $m^{O(t^2)}$ time on $t$-boundaried graphs
when the above-mentioned construction is also given in input.

\subsection{Definitions}

Consider an instance $G$ of the {\stackmst} problem with $R(G)$ and $B(G)$ denoting
the set of red and blue edges, respectively, and with  cost function $c: R(G) \to \mathbb{R}_{+}$
on the set of red edges. As usual, denote by $c_{1}, c_{2}, \dots, c_{k}$ the different values taken
by $c$, in increasing order, and let $c_{0}:=0$.\medskip

For two distinct vertices $u,v \in V(G)$ of $G$ and
a subset $F \subseteq B(G)$ of blue edges, the sets $\paths{G}{F}{u}{v}$ and
$\redpaths{G}{F}{u}{v}$ are defined exactly as in Section~\ref{sec-def}, that is,
$\paths{G}{F}{u}{v}$ is the set of $uv$-paths in $(V(G), R(G) \cup F)$, while
$\redpaths{G}{F}{u}{v}$ denotes the subset of those paths that contain at least one red edge.
The corresponding quantities $\bneck{G}{F}{u}{v}$ and $\redbneck{G}{F}{u}{v}$
are also defined as before, that is,
$\bneck{G}{F}{u}{v}$ is the minimum of $\mc(P)$ over every path $P \in \paths{G}{F}{u}{v}$,
with $\bneck{G}{F}{u}{v}:=c_{k}$ if there is no such path, and
$\redbneck{G}{F}{u}{v}$ is defined in the same way but with respect to
$\redpaths{G}{F}{u}{v}$.

Now let us further assume the instance $G$ is a $t$-boundaried graph, and let us consider
two distinct boundary vertices $a,b \in \partial(G)$. An
$ab$-path of $G$ is said to be {\sl internal}
if the only boundary vertices of $G$ it includes are $a$ and $b$.
For $F \subseteq B(G)$,
the sets $\intpaths{G}{F}{a}{b}$ and $\intredpaths{G}{F}{a}{b}$ are defined
as $\paths{G}{F}{a}{b}$ and $\redpaths{G}{F}{a}{b}$, respectively, but with the additional
requirement that the $ab$-paths under consideration are internal $ab$-paths. The quantities
$\intbneck{G}{F}{a}{b}$ and $\intredbneck{G}{F}{a}{b}$ are defined with
respect to $\intpaths{G}{F}{a}{b}$ and $\intredpaths{G}{F}{a}{b}$, respectively, as expected.

For clarity, in what follows we will use the following convention:
the letters $a$ and $b$ will always denote vertices in the boundary
of $G$, while $u$ and $v$ will be used for arbitrary (possibly non-boundary) vertices of $G$.

A {\sl $k$-graph on the boundary of $G$}, or simply {\sl $k$-graph} when $G$ is clear from the context,
is a triple $I=(K, f, g)$ where
$K$ is a complete graph with vertex set $\partial(G)$, and $f:E(K) \to [k]$ and
$g:E(K) \to [k]$ are two functions assigning weights in $[k]$ to the edges of $K$.
(Let us recall that, by our convention, $[k]$ denotes the set $\{0, 1, \dots, k\}$.)
We say that a subset $F\subseteq B(G)$ of blue edges of $G$ {\sl realizes}
a $k$-graph $I=(K, f, g)$ if $F$ is acyclic,
and for every two distinct vertices $a,b\in \partial(G)$
we have $\intbneck{G}{F}{a}{b}=c_{f(ab)}$ (thus there is no condition on $g$).
The $k$-graph $I$ is said to be {\sl realizable} in $G$ if there exists such a subset
$F$ of blue edges.
Notice that this is a direct extension of the notion of realizability introduced
in Section~\ref{sec-def} for series-parallel graphs.
We define $G + I$ as the ($t$-boundaried) graph obtained from $G$
by adding, for every two distinct vertices $a,b\in \partial(G)$,
a red edge connecting $a$ and $b$ with cost $c_{g(ab)}$. We let
$\OPTtw_{I}(G)$ be defined as follows:
$$
\OPTtw_{I}(G) := \max\left.\left\{ \sum_{uv \in F} \redbneck{G + I}{F}{u}{v}
\,\right| F\subseteq B(G), F \textrm{ realizes } I  \textrm{ in } G \right\}.
$$
In cases where $OPT_{I}(G)$ is undefined (that is, $I$ is not realizable),
then we set $OPT_{I}(G)=-\infty$.

With these definitions, the dynamic program that will be used is a straightforward generalization
of the series-parallel case: We store for every $t$-boundaried graph $H$ appearing in the construction of our $t$-boundaried input graph $G$ the value $\OPTtw_{I}(H)$ for every $k$-graph $I$, together with a corresponding  optimal acyclic subset $F$ of blue edges (if $\OPTtw_{I}(H) > -\infty$).
The value returned by the dynamic programming solution is then the maximum
of $\OPTtw_{I}(G)$ over all $k$-graphs $I$, and a corresponding acyclic subset of blue edges
of $G$ is returned.
By Lemma~\ref{lem-opt-for-fixed-blue}, this is the maximum revenue achievable by the leader.

Now we consider the five operators appearing in the definition of $t$-boundaried graphs,
and show for each of them  how to compute $\OPTtw_{I}(G)$
from already computed values when $G$ results from the application of the operator.

\subsection{The null operator $\varnothing$}

We begin with the null operator $\varnothing$ that creates a new graph $G$ with $t$ isolated
boundary vertices labeled $1,\ldots,t$.
Consider an arbitrary $k$-graph $I=(K,f,g)$ on the boundary of $G$. If $f(ab) < k$
for some edge $ab \in E(K)$, then $I$ is not realizable in $G$, because there is no internal
$ab$-path in $G$. Thus we set $OPT_{I}(G) :=-\infty$ in this case.

If, on the other hand, $f(e) = k$ for every $e\in E(K)$, then the
subset $F=\varnothing$ of blue edges of $G$ realizes $I$, and it is of course the only one
since $B(G) = \varnothing$. Hence we let
$OPT_{I}(G):=0$ (associated with the set $F=\varnothing$).

\subsection{The binary operator $\oplus$}

The $\oplus$ operator is very similar to a parallel-composition of series-parallel graphs.
Suppose that $G = G_1\oplus G_2$, and let $I=(K,f,g)$ be an arbitrary $k$-graph on the boundary of $G$.
We extend the notion of parallel-compatibility from Section~\ref{sec-parallel} as follows:
If $I_{1}=(K_{1},f_{1},g_{1})$ and $I_{2}=(K_{2},f_{2},g_{2})$ are two $k$-graphs,
then we say that
the pair $(I_{1},I_{2})$ is {\sl {$\oplus$-compatible}} with $I$ if
$I_{i}$ ($i=1,2$) is realizable in $G_{i}$, and moreover
the following five conditions are satisfied for every $e \in E(K)$:
\begin{enumerate}
  \item[(1)]  at least one of $f_{1}(e)$ and $f_{2}(e)$ is non-zero;
  \item[(2)]  $f(e)=\min\{f_{1}(e), f_{2}(e)\}$;
  \item[(3)]  $g_{1}(e) = \min\{g(e),f_{2}(e)\}$;
  \item[(4)]  $g_{2}(e) = \min\{g(e),f_{1}(e)\}$, and
  \item[(5)] for every cycle $C$ in $K$,
  there exists $i\in \{1,2\}$ such that
  $f_{i}(e) > 0$ for every $e\in E(C)$.
\end{enumerate}

\noindent
Our goal is to compute $\OPTtw_{I}(G)$ based on values already computed for $G_{1}$ and $G_{2}$.
This is achieved by the following lemma.

\begin{lemma}
\label{lem-opt-oplus}
Assume that $G$, $I$, $G_{1}$ and $G_{2}$ are as above, and suppose further that $I$ is realizable
in $G$. Then
$$
\OPTtw_{I}(G) = \max \{ \OPTtw_{I_1}(G_1) + \OPTtw_{I_2}(G_2) \mid
(I_{1},I_{2})\text{\ is } \oplus \text{-compatible with } I \}.
$$
\end{lemma}

(Let us remark that, if $I$ is not realizable in $G$, then we trivially have
$\OPTtw_{I}(G) = -\infty$.)
The proof of Lemma~\ref{lem-opt-oplus} is a generalization of the proof of
Lemma~\ref{lem-opt-parallel} for parallel compositions and consists
of a few steps.  First we prove the following lemma, which is similar
to Lemma~\ref{lem-parallel-F}.

\begin{lemma}
\label{lem-oplus-F}
Suppose that  $I_{i} = (K_{i}, f_{i}, g_{i})$ is a $k$-graph realized in $G_{i}$
by a subset $F_{i} \subseteq B(G_{i})$ of blue edges of $G_{i}$, for $i=1,2$,
and assume further that $(I_{1},I_{2})$ is $\oplus$-compatible with $I$.
Then $F:= F_{1} \cup F_{2}$ realizes $I$ in $G$.
\end{lemma}
\begin{proof}
We have to prove that $F$ is acyclic and that $\intbneck{G}{F}{a}{b}=c_{f(ab)}$ for every edge $ab\in E(K)$.

First, suppose that $(V(G), F)$ contains a cycle $C$.
Since $F_1$ and $F_{2}$ are both acyclic, $C$ includes at least two distinct boundary vertices  $a$ and $b$, and moreover
$E(G_1) \cap E(C)$, $E(G_2) \cap E(C)$ are both non-empty. If $a$ and $b$ are the only boundary vertices in $C$ then
 there is an $ab$-path in $(V(G), F_{1})$ and an $ab$-path  in $(V(G), F_{2})$,
implying  that $f_{1}(ab) = f_{2}(ab) = 0$, which contradicts condition (1)
from the definition of $\oplus$-compatibility.

If, on the other hand, $C$ contains at least three boundary vertices,
choose an orientation of $C$ and an arbitrary vertex $a_{1} \in \partial(G) \cap V(C)$,
and enumerate the vertices in $\partial(G) \cap V(C)$
as $a_{1}, a_{2}, \dots, a_{p}$ according to the order in which they appear when
walking on $C$ from $a_{1}$ in the chosen orientation.
By condition (5), there is an index $j\in \{1,2\}$ such that
$f_{j}(a_{i}a_{i+1}) > 0$ for every $i\in \{1, \dots, p\}$
(taking indices modulo $p$). We may assume without loss of generality that
this is the case for $j=1$.

For every $i\in \{1, \dots, p\}$,
the (oriented) path from $a_{i}$ to $a_{i+1}$ in $C$
is a subgraph of $(V(G), F_{1})$ or $(V(G), F_{2})$, since it does not
contain other boundary vertices than $a_{i}$ and $a_{i+1}$.
This path cannot be a subgraph of $(V(G), F_{1})$ since $f_{1}(a_{i}a_{i+1}) > 0$,
hence it is contained in $(V(G), F_{2})$. However, it follows then that $C$ itself
is a subgraph of $(V(G), F_{2})$, which contradicts the fact that $F_{2}$ is acyclic.
Therefore, $F$ must be acyclic.

Now, consider two distinct vertices $a,b\in \partial(G)$.
Clearly
$\intpaths{G_{1}}{F_{1}}{a}{b} \cup \intpaths{G_{2}}{F_{2}}{a}{b} \subseteq \intpaths{G}{F}{a}{b}$.
By definition,
 each path  $P\in \intpaths{G}{F}{a}{b}$ has no other boundary vertices
than $a$ and $b$, hence $P$ is included in
either $\intpaths{G_{1}}{F_{1}}{a}{b}$ or $\intpaths{G_{2}}{F_{2}}{a}{b}$.
It follows that
$\intpaths{G}{F}{a}{b} = \intpaths{G_{1}}{F_{1}}{a}{b} \cup \intpaths{G_{2}}{F_{2}}{a}{b}$.
This in turn implies
$\intbneck{G}{F}{a}{b} =
\min\{ \intbneck{G_{1}}{F_{1}}{a}{b}, \intbneck{G_{2}}{F_{2}}{a}{b}\}
= \min\{c_{f_{1}(ab)},c_{f_{2}(ab)}\}$, which is equal to $c_{f(ab)}$
by condition (2).
\begin{flushright}\qed\end{flushright}\end{proof}

The next lemma is the analogue of Lemma~\ref{lem-parallel-bneck} from Section~\ref{sec-parallel}.
\begin{lemma}
\label{lem-oplus-bneck}
Let $I_{1} = (K_{1}, f_{1}, g_{1})$, $I_{2} = (K_{2}, f_{2}, g_{2})$, $F_{1}$, $F_{2}$, and
$F$ be as in Lemma~\ref{lem-oplus-F}. Then, for $i=1,2$, and every
edge $uv \in F_{i}$, we have
$$
\redbneck{G+I}{F}{u}{v} = \redbneck{G_{i}+ I_{i}}{F_{i}}{u}{v}.
$$
\end{lemma}
\begin{proof}
We prove the statement for $i=1$, the case $i=2$ follows by symmetry.

For every two distinct vertices $a,b \in  \partial(G)$, let $e^{ab}$ and $e_1^{ab}$ be the additional red edges in $G + I$ and $G + I_{1}$, respectively, between the boundary vertices $a$ and $b$.

Let $uv \in F_{1}$. We first show:

\begin{claim}
\label{claim-oplus-leq}
$\redbneck{G+I}{F}{u}{v} \leq \redbneck{G+I_{1}}{F_{1}}{u}{v}$.
\end{claim}
\begin{proof}
If $\redpaths{G+I_{1}}{F_{1}}{u}{v}$ is empty then trivially
$\redbneck{G+I}{F}{u}{v} \leq c_{k} = \redbneck{G+I_{1}}{F_{1}}{u}{v}$, thus we may assume
$\redpaths{G+I_{1}}{F_{1}}{u}{v} \neq \varnothing$.

Let $P_{1}$ be a path in $\redpaths{G+I_{1}}{F_{1}}{u}{v}$
with $\mc(P_{1}) = \redbneck{G+I_{1}}{F_{1}}{u}{v}$ and minimizing its length.
We will show the existence of a path $P$ in
in $\redpaths{G+I}{F}{u}{v}$ with $\mc(P) \leq \mc(P_{1})$. Since
$\redbneck{G+I}{F}{u}{v} \leq \mc(P)$, this will imply the claim.

If $P_{1}$ includes at most one boundary vertex, then
$P_{1} \in \redpaths{G+I}{F}{u}{v}$ and we are done. Hence we may assume that
$P_{1}$ includes at least two boundary vertices.
Enumerate
the boundary vertices that are included in $P_{1}$ as $a_{1}, \dots, a_{p}$, in the order in which they
appear when going from $u$ to $v$.
Let $X$ be the set of indices $i\in \{1, \dots, p-1\}$ such that
the subpath $a_{i}P_{1}a_{i+1}$ of $P_{1}$ consists of the edge $e_{1}^{a_{i}a_{i+1}}$.
The latter edges are exactly the edges of $P_{1}$ that do no exist in $G+I$.
(Note that there could be none, that is, $X$ could be empty.)

For every $i \in X$, we have by condition (3) from the definition of $\oplus$-compatibility
that $g_{1}(a_{i}a_{i+1})$ is equal to the minimum of
$g(a_{i}a_{i+1})$ and $f_{2}(a_{i}a_{i+1})$. We define an internal $a_{i}a_{i+1}$-path $Q_{i}$
as follows: If $g_{1}(a_{i}a_{i+1}) = g(a_{i}a_{i+1})$, then $Q_{i}$ consists simply of the
edge $e^{a_{i}a_{i+1}}$. Otherwise, we let $Q_{i}$ be a path in
$\intpaths{G_{2}}{F_{2}}{a_{i}}{a_{i+1}}$
with $\mc(Q_{i}) = f_{2}(a_{i}a_{i+1}) = g_{1}(a_{i}a_{i+1})$.
(Observe that such a path exists since $F_{2}$ realizes $I_{2}$ in $G_{2}$.)
In both cases, $Q_{i}$ is a path which is a subgraph of $G+I$.

We claim that, for every $i,j \in X$ with $i < j$, the path $Q_{i}$ is internally disjoint
from $Q_{j}$ (that is, the only vertex they may have in common is $a_{i+1}$ provided $j = i + 1$).
Arguing by contradiction, assume otherwise. Then
the union of $Q_{i}$ and $Q_{j}$ contains an internal $a_{i}a_{j+1}$-path $R$,
and this path satisfies
 $\mc(R) \leq \max \{\mc(Q_{i}), \mc(Q_{j})\} = \max\{ g_{1}(a_{i}a_{i+1}),
 g_{1}(a_{j}a_{j+1})\} \leq \mc(P_{1})$. But then it follows from condition (3)
that $g_{1}(a_{i}a_{j+1}) \leq \mc(R) \leq \mc(P_{1})$.
Thus, replacing the $a_{i}P_{1}a_{j+1}$ subpath
of $P_{1}$ with the edge $e_{1}^{a_{i}a_{j+1}}$ gives a path $P'_{1}$ in
$\redpaths{G+I_{1}}{F_{1}}{u}{v}$  with
$\mc(P'_{1}) \leq \mc(P_{1})=\redbneck{G+I_{1}}{F_{1}}{u}{v}$ (and hence with
$\mc(P'_{1}) =\redbneck{G+I_{1}}{F_{1}}{u}{v}$), which is shorter than $P_{1}$, a contradiction.

For each $i\in X$, the path $Q_{i}$ has no other vertex in common with $P_{1}$ than its two endpoints
(since $Q_{i}$ is an internal $a_{i}a_{i+1}$-path from $G_{2}$).
Relying on the fact that the $Q_{i}$'s are pairwise internally disjoint, we let $P$
be the path obtained from $P_{1}$ by replacing, for every $i\in X$,
the edge $e_{1}^{a_{i}a_{i+1}}$ with the path $Q_{i}$. The path $P$ must contain
at least one red edge, because otherwise $P+uv$ would be a cycle in $(V(G), F)$,
contradicting Lemma~\ref{lem-oplus-F}. Thus $P$ is in
$\redpaths{G+I}{F}{u}{v}$. Moreover, by our choice of the $Q_{i}$'s, we have
$\mc(P) \leq \mc(P_{1})$, as desired.
\begin{flushright}\qed\end{flushright}\end{proof}

Conversely, we prove:

\begin{claim}
\label{claim-oplus-geq}
$\redbneck{G+I}{F}{u}{v} \geq \redbneck{G+I_{1}}{F_{1}}{u}{v}$.
\end{claim}
\begin{proof}
If $\redpaths{G+I}{F}{u}{v}$ is empty then
$\redbneck{G+I}{F}{u}{v} = c_{k} \geq \redbneck{G+I_{1}}{F_{1}}{u}{v}$, thus we may suppose that
$\redpaths{G+I}{F}{u}{v}$ is not empty.

We have to show that $\mc(P) \geq \redbneck{G+I_{1}}{F_{1}}{u}{v}$
for every $P \in \redpaths{G+I}{F}{u}{v}$. Consider such a path $P$.
If $P$ includes at most one boundary vertex, then $P \in \redpaths{G+I}{F}{u}{v}$
and we are done. So assume $P$ contains at least two boundary vertices,
and enumerate them as $a_{1}, \dots, a_{p}$ as in the proof of the previous claim.

For every $i \in \{1, \dots, p-1\}$, the  subpath $Q_{i}:=a_{i}Pa_{i+1}$ of $P$
is either in $\intpaths{G_{1}}{F_{1}}{a_{i}}{a_{i+1}}$, or in $\intpaths{G_{2}}{F_{2}}{a_{i}}{a_{i+1}}$,
or consists of the edge $e^{a_{i}a_{i+1}}$. Observe that, in the second case, we have
$g_{1}(a_{i}a_{i+1}) \leq f_{2}(a_{i}a_{i+1}) \leq \mc(Q_{i})$ by condition (3),
and in the last case $g_{1}(a_{i}a_{i+1}) \leq g(a_{i}a_{i+1}) = \mc(Q_{i})$ by the same condition.
Hence, if for every $i \in \{1, \dots, p-1\}$ such that $Q_{i}\notin \intpaths{G_{1}}{F_{1}}{a_{i}}{a_{i+1}}$, we replace the subpath $Q_{i}$ of $P$ with the edge $e^{a_{i}a_{i+1}}$, we obtain
a path $P_{1}$ which is in $\redpaths{G+I_{1}}{F_{1}}{u}{v}$ and which satisfies
$\mc(P_{1}) \leq \mc(P)$. Since $\redbneck{G+I_{1}}{F_{1}}{u}{v} \leq \mc(P_{1})$, this completes
the proof.
\begin{flushright}\qed\end{flushright}\end{proof}

Lemma~\ref{lem-oplus-bneck} follows from Claims~\ref{claim-oplus-leq} and~\ref{claim-oplus-geq}.
\begin{flushright}\qed\end{flushright}\end{proof}

We may now turn to the proof of Lemma~\ref{lem-opt-oplus}.

\begin{proof}[Proof of Lemma~\ref{lem-opt-oplus}]
We first show:
\begin{claim}
\label{claim-opt-oplus-exists}
There exist $k$-graphs $I_{1}$ and $I_{2}$ such that $(I_{1},I_{2})$ is $\oplus$-compatible with $I$
and $\OPTtw_I(G) \leq \OPTtw_{I_1}(G_1) + \OPTtw_{I_2}(G_2)$.
\end{claim}
\begin{proof}
Let $F\subseteq B(G)$ be a subset of blue edges realizing
$I$ in $G$ such that
$$
OPT_I(G)  = \sum_{uv \in F} \redbneck{G+I}{F}{u}{v}.
$$

For $i=1,2$, let $F_{i}:=F \cap E(G_{i})$, and let $I_{i}=(K_{i}, f_{i}, g_{i})$
be the $k$-graph obtained by letting, for every $ab \in E(K)$,
$f_{i}(ab)$ be the index $j\in [k]$ such that $c_{j} = \intbneck{G_{i}}{F_{i}}{a}{b}$,
and $g_{i}(ab) := \min\{g(ab), f_{i+1}(ab)\}$
(indices are taken modulo 2).
Observe that $F_{i}$ realizes $I_{i}$ in $G_{i}$, for $i=1,2$.

Let us show that $(I_{1}, I_{2})$ is $\oplus$-compatible with $I$.
Condition (1) from the definition of $\oplus$-compatibility is satisfied
because otherwise the graph $(V(G), F)$ would have a cycle.
It should be clear from the definitions of $I_{1}$ and $I_{2}$
that conditions (2), (3) and (4) are also
satisfied. Hence, it remains to check condition (5).
Arguing by contradiction, let us assume it is not satisfied, that is, that
there exists a cycle in $K$ containing two edges $e$ and $e'$
such that $f_{1}(e)=0$ and $f_{2}(e')=0$. Such a cycle is said to be {\sl bad}.

Let $C$ be a shortest bad cycle in $K$.
Consider an arbitrary orientation of $C$ and enumerate the vertices of $C$
as $a_{1}, a_{2}, \dots, a_{p}$, in order.
By condition (1), for every $i\in \{1, \dots, p\}$ there is a {\em unique}
index $j\in \{1,2\}$ such that $f_{j}(a_{i}a_{i+1})=0$ (indices are taken modulo $p$);
let $\varphi(i)$ denote this index.

Let $Q_{i}$ be the (unique) $a_{i}a_{i+1}$-path in $(V(G_{\varphi(i)}), F_{\varphi(i)})$,
for every $i\in \{1, \dots, p\}$. Note that $Q_{i}$ is necessarily an {\em internal}
$a_{i}a_{i+1}$-path, that is, $Q_{i}$ does not contain any other boundary vertex than
$a_{i}$ and $a_{i+1}$.
We claim that the $Q_{i}$'s are pairwise internally
disjoint. Assume this is not the case, that is, that $Q_{i}$ and $Q_{j}$ share
an internal vertex $v$ for some $i,j \in \{1, \dots, p\}$ with $i < j$.
Since $v$ is not a boundary vertex, we must have  $\varphi(i) = \varphi(j)$.
For simplicity, assume without loss of generality that $\varphi(i) = 1$.
For every $a\in \{a_{i}, a_{i+1}\}$ and $b\in \{a_{j}, a_{j+1}\}$ with $a\neq b$,
there is an internal $ab$-path in the union of $Q_{i}$ and $Q_{j}$, implying
$f_{1}(ab) = 0$. If $|C| \geq 4$ then $a$ and $b$ can be chosen such that
$ab$ is not an edge of $C$. Then the chord $ab$ splits $C$ into two cycles,
at least one of which is bad. However, this implies that there is a bad cycle in $K$ that is shorter than $C$, a contradiction. If $|C| = 3$, then it follows that
$f_{1}(a_{1}a_{2}) = f_{1}(a_{2}a_{3}) = f_{1}(a_{3}a_{1}) = 0$. But we also
have $f_{2}(a_{i}a_{i+1}) = 0$ for some $i\in \{1,2,3\}$ since $C$ is bad,
which contradicts condition (1). Since in both cases we reach a contradiction,
we deduce that the $Q_{i}$'s must be pairwise internally disjoint.

Let $C'$ be obtained from the cycle $C$ by replacing each edge $a_{i}a_{i+1}$ ($i\in \{1, \dots, p\}$)
with the path $Q_{i}$. Then $C'$ is a cycle, since $Q_{i}$ and $Q_{j}$ are internally disjoint
for every $i < j$, and is a subgraph of $(V(G), F)$, contradicting the fact that $F$ is acyclic.
Therefore, there cannot be any bad cycle in $K$, and condition (5) holds.

Now that we know that $(I_{1}, I_{2})$ is $\oplus$-compatible with $I$,
we may apply Lemma~\ref{lem-oplus-bneck}:
\begin{align*}
OPT_I(G)&= \sum_{uv \in F} \redbneck{G + I}{F}{u}{v} \\
 &= \sum_{uv \in F_{1}} \redbneck{G + I_{1}}{F_{1}}{u}{v}
+  \sum_{uv \in F_{2}} \redbneck{G + I_{2}}{F_{2}}{u}{v} \\
&\leq OPT_{I_1}(G_1) +OPT_{I_2}(G_2).
\end{align*}
\begin{flushright}\qed\end{flushright}\end{proof}

Next we prove:

\begin{claim}
\label{claim-opt-oplus-forall}
$OPT_{I}(G)  \geq OPT_{I_1}(G_1)  +OPT_{I_2}(G_2) $
holds for every $I_{1}, I_{2}$ such that $(I_{1},I_{2})$ is  {\sl {$\oplus$-compatible}} with $I$.
\end{claim}
\begin{proof}
Suppose that $(I_{1},I_{2})$ is  {\sl {$\oplus$-compatible}} with $I$.
Let $F_{i}\subseteq B(G_{i})$
($i=1,2$) be a subset of blue edges of $G_{i}$ realizing $I_{i}$
such that
$$
OPT_{I_i}(G_i)
= \sum_{uv \in F_{i}} \redbneck{G +I_{i}}{F_{i}}{u}{v}.
$$
By Lemma~\ref{lem-oplus-F}, $F:= F_{1} \cup F_{2}$ realizes
$I$ in $G$. By Lemma~\ref{lem-oplus-bneck}, we have:
\begin{align*}
OPT_{I}(G)  &\geq \sum_{uv \in F} \redbneck{G + I}{F}{u}{v} \\
&= \sum_{uv \in F_{1}} \redbneck{G + I_{1}}{F_{1}}{u}{v}
+  \sum_{uv \in F_{2}} \redbneck{G + I_{2}}{F_{2}}{u}{v} \\
&= OPT_{I_1}(G_1) +OPT_{I_2}(G_2).
\end{align*}
\begin{flushright}\qed\end{flushright}\end{proof}

Lemma~\ref{lem-opt-oplus} follows from Claims~\ref{claim-opt-oplus-exists} and~\ref{claim-opt-oplus-forall}.
\begin{flushright}\qed\end{flushright}\end{proof}

\subsection{The unary operator $\eta$}
Suppose that $G=\eta(G')$, that is, that $G$ is obtained from $G'$ by adding
a new isolated boundary vertex $\etanode{b}$ and labeling it $1$. Thus
the vertex $\etanode{a}$ with label $1$ in the boundary of $G'$ is no longer
a boundary vertex in $G$.

The graphs $G$ and $G'$ have exactly the same set of edges.
However, an $ab$-path between two distinct boundary vertices $a,b\in \partial(G) \cap \partial(G')$ that goes through $\etanode{a}$ is not an {\em internal} path in $G'$, but could be in $G$
(if the path does not contain any other boundary vertex).
This leads us to the following definition. Let $I=(K, f, g)$ be an arbitrary
$k$-graph on the boundary of $G$. Then a $k$-graph $I'=(K', f', g')$ on
the boundary of $G'$ is {\sl {$\eta$-compatible}} with $I$ if $I'$ is realizable in $G'$ and,
for every two distinct vertices $a,b\in \partial(G) \cap \partial(G')$,
the following four conditions hold:
\begin{enumerate}
\item[(1)] $f(ab)=\min \big \{f'(ab),\max\{f'(a\etanode{a}),f'(\etanode{a}b)\}\big\}$;
\item[(2)]  $g'(ab)=\min \big \{g(ab),\max\{g(a\etanode{b}),g(\etanode{b}b)\}\big\}$;
\item[(3)] $f(a\etanode{b})=k$, and
\item [(4)]$g'(a\etanode{a})=k$.
\end{enumerate}

\begin{lemma}
\label{lem-opt-eta}
Assume that $G$, $I$, and $G'$ are as above, and suppose further that $I$ is realizable
in $G$. Then
$$
\OPTtw_{I}(G) = \max \{ \OPTtw_{I'}(G) \mid
I' \text{\ is } \eta \text{-compatible with } I \}.
$$
\end{lemma}

(Again, if $I$ is not realizable in $G$, then trivially $\OPTtw_{I}(G) = -\infty$.)
The proof of Lemma~\ref{lem-opt-eta} is split into a few lemmas, as in the previous section.
We begin with the following lemma.

\begin{lemma}
\label{lem-eta-F}
Suppose that $F' \subseteq B(G')$ realizes a $k$-graph $I'=(K',f',g')$ in $G'$ which is
$\eta$-compatible with $I$. Then $F:= F'$ realizes $I$ in $G$.
\end{lemma}
\begin{proof}
Since $F'$ realizes $I'$ in $G'$, the set $F=F'$ is acyclic, we are left with proving
that $\intbneck{G}{F}{a}{b}=c_{f(ab)}$ for every edge $ab\in E(K)$. Let thus $ab$
be an arbitrary edge in $E(K)$.

First suppose that $a$ or $b$ is equal to $\etanode{b}$, say without loss of generality
$b=\etanode{b}$. Since $b$ is an isolated vertex of $G$,
we have $\intpaths{G}{F}{a}{b}=\varnothing$, and thus $\intbneck{G}{F}{a}{b}=c_k$.
We also have $f(ab)=k$ by condition (3) from the definition of $\eta$-compatibility; hence $\intbneck{G}{F}{a}{b}=c_{f(ab)}$ as desired.

Next suppose that $a,b \neq \etanode{b}$.
For every path $P \in \intpaths{G}{F}{a}{b}$, either $P$ includes the vertex $\etanode{a}$ or not.
If $\etanode{a}\not \in V(P)$, then $P$ is also an internal $ab$-path in $G'$. If $\etanode{a}\in V(P)$, then $P$ is not internal in $G'$ but $P$ is the concatenation of an internal $a\etanode{a}$-path
$P_{1}$ in $G'$ with an internal $\etanode{a}b$-path $P_{2}$ in $G'$, and thus
$\mc(P) =  \max\{ \mc(P_1), \mc(P_2)\}$. It follows that
$$
\intbneck{G}{F}{a}{b} \geq \min \big \{   \intbneck{G'}{F}{a}{b},  \max\{ \intbneck{G'}{F}{a}{\etanode{a}},  \intbneck{G'}{F}{\etanode{a}}{b}  \} \big\}.
$$
Let us show that the reverse inequality also holds.
This is easy to see
if $\intbneck{G'}{F}{a}{b} \leq \max\{ \intbneck{G'}{F}{a}{\etanode{a}},  \intbneck{G'}{F}{\etanode{a}}{b}  \} \big\}$, since
every path in $\intpaths{G'}{F}{a}{b}$ is included in $\intpaths{G}{F}{a}{b}$, implying
$\intbneck{G}{F}{a}{b} \leq \intbneck{G'}{F}{a}{b}$.

Let us thus assume $\intbneck{G'}{F}{a}{b} > \max\{ \intbneck{G'}{F}{a}{\etanode{a}},  \intbneck{G'}{F}{\etanode{a}}{b}  \} \big\}$, and let
$P_{1} \in \intpaths{G'}{F}{a}{\etanode{a}}$
and $P_{2} \in \intpaths{G'}{F}{\etanode{a}}{b}$
be such that $\mc(P_{1}) = \intbneck{G'}{F}{a}{\etanode{a}}$
and $\mc(P_{2}) = \intbneck{G'}{F}{\etanode{a}}{b}$.
Then $P_{1}$ and $P_{2}$ cannot
have another vertex in common than $\etanode{a}$, because otherwise their union would
contain an $ab$-path $P$ avoiding $\etanode{a}$,  which is thus in
$\intpaths{G'}{F}{a}{b}$. This in turn implies
$\intbneck{G'}{F}{a}{b} \leq \mc(P) \leq \max\{ \mc(P_1), \mc(P_2)\} = \max\{ \intbneck{G'}{F}{a}{\etanode{a}},  \intbneck{G'}{F}{\etanode{a}}{b}  \} \big\}$, which contradicts our hypothesis.
Hence, $V(P_{1}) \cap V(P_{2}) = \{\etanode{a}\}$, and
the concatenation of $P_{1}$ and $P_{2}$ gives an $ab$-path
$P$ which is internal in $G$ (but not in $G'$), and which is thus included in
$\intpaths{G}{F}{a}{b}$. This implies
$\intbneck{G}{F}{a}{b} \leq \mc(P) = \max\{ \mc(P_1), \mc(P_2)\} = \max\{ \intbneck{G'}{F}{a}{\etanode{a}},  \intbneck{G'}{F}{\etanode{a}}{b}  \}$, as desired.

Therefore,
\begin{align*}
\intbneck{G}{F}{a}{b} &= \min \big \{   \intbneck{G'}{F}{a}{b},  \max\{ \intbneck{G'}{F}{a}{\etanode{a}},  \intbneck{G'}{F}{\etanode{a}}{b}  \} \big\} \\
&= \min \big \{   c_{f'(ab)},  \max\{ c_{f'(a\etanode{a})}, c_{f'(\etanode{a}b)}  \} \big\},
\end{align*}
which is equal to $c_{f(ab)}$ by condition (1).
\begin{flushright}\qed\end{flushright}\end{proof}

\begin{lemma}
\label{lem-eta-bneck}
Let $I' = (K', f', g')$ and $F'$ be as in Lemma~\ref{lem-eta-F}, and let $F:= F'$.
Then, for every edge $uv \in F$,
$$
\redbneck{G+I}{F}{u}{v} = \redbneck{G'+I'}{F}{u}{v}.
$$
\end{lemma}
\begin{proof}
For every $ab \in E(K)$, let $e^{ab}$ be the extra red edge in $G+I$
between the boundary vertices $a$ and $b$. Similarly,
for every $ab \in E(K')$, let $e'^{ab}$ be the extra red edge in $G'+I'$
between the boundary vertices $a$ and $b$.

Let $uv \in F$. The proof consists of three claims.

\begin{claim}
\label{claim-eta-empty}
If $\redpaths{G+I}{F}{u}{v}=\varnothing$ or
$\redpaths{G' + I'}{F}{u}{v}=\varnothing$ then
$\redbneck{G + I}{F}{u}{v} = \redbneck{G' + I'}{F}{u}{v} = c_k$.
\end{claim}

\begin{proof}

First suppose that $\redpaths{G + I}{F}{u}{v}=\varnothing$. Then $\redbneck{G + I}{F}{u}{v} = c_k$ by definition. If $\redpaths{G'+I'}{F}{u}{v}=\varnothing$ as well then $\redbneck{G + I}{F}{u}{v} = \redbneck{G' + I'}{F}{u}{v} = c_k$, and we are done. Let us thus assume that $\redpaths{G'+I'}{F}{u}{v}$
is not empty. Every path $P\in \redpaths{G' + I'}{F}{u}{v}$ contains an extra red edge
of the form $e^{'ab}$ with $a$ or $b$ being equal to $\etanode{a}$,
since $\redpaths{G + I}{F}{u}{v}=\varnothing$. The cost of this extra edge is
$c_{g'(ab)}$, which is equal to $c_{k}$ by condition (4). It follows that
$\mc(P) = c_k$, and hence $\redbneck{G' + I'}{F}{u}{v} = c_k$, as desired.

Now assume that $\redpaths{G' + I'}{F}{u}{v}=\varnothing$. We show that this implies
$\redpaths{G + I}{F}{u}{v}=\varnothing$ as well, which reduces this case to the case treated above.
Arguing by contradiction, suppose that
$\redpaths{G + I}{F}{u}{v}\neq \varnothing$, and let $P\in \redpaths{G + I}{F}{u}{v}$.
Since $\redpaths{G' + I'}{F}{u}{v}=\varnothing$, the path $P$ must contain
the vertex $\etanode{b}$. The two edges of $P$ incident to $\etanode{b}$
are extra red edges of the form $e^{a\etanode{b}}$ and $e^{b\etanode{b}}$, respectively,
with $a, b \in \partial(G) \cap \partial(G')$ and $a\neq b$. However, replacing the subpath
of $P$ consisting of these two edges with the edge $e^{ab}$ gives a path in
$\redpaths{G + I}{F}{u}{v}$ avoiding $\etanode{b}$, implying that
$\redpaths{G' + I'}{F}{u}{v}$ is not empty, a contradiction. The claim follows.
\begin{flushright}\qed\end{flushright}\end{proof}

\begin{claim}
\label{claim-eta-leq}
If $\redpaths{G+I}{F}{u}{v}\neq \varnothing$ and
$\redpaths{G'+I'}{F}{u}{v}\neq \varnothing$ then $\redbneck{G+I}{F}{u}{v} \leq \redbneck{G' + I'}{F}{u}{v}$.
\end{claim}
\begin{proof}
We have to show that $\redbneck{G+ I}{F}{u}{v} \leq \mc(P')$
for every path $P' \in \redpaths{G'+I'}{F}{u}{v}$. Consider such a path $P'$.
If $P'$ contains no extra red edge
(that is, a red edge of the form $e'^{ab}$ with $a,b \in \partial(G')$),
then $P' \in \redpaths{G + I}{F}{u}{v}$, and $\redbneck{G+ I}{F}{u}{v} \leq \mc(P')$
holds. Thus we may assume that $P'$ contains at least one such edge.

If $P'$ includes an edge of the form $e'^{ab}$ with $a$ or $b$ being equal to $\etanode{a}$, then
this edge has cost $c_{g'(ab)}=c_{k}$ by condition (4), implying $\mc(P')=c_{k}$, and thus we have
$\redbneck{G+ I}{F}{u}{v} \leq c_{k} = \mc(P')$. Hence we may assume that $P'$ has no such edge.

Let $H$ be the subgraph of $G+I$ obtained from $P'$ as follows: for each
each extra red edge $e'^{ab}$ included in $P'$, replace $e'^{ab}$ with
$e^{ab}$ if $g'(ab)=g(ab)$, with the path consisting of the two edges $e^{a\etanode{b}}$,
$e^{\etanode{b}b}$ otherwise. Note that $H$ is connected but is not necessarily a path, since the
vertex $\etanode{b}$ could have degree more than $2$ in $H$. On the other hand, we have
$\mc(H) = \mc(P')$ by condition (2). Also, note that every $uv$-path in $H$ contains at least one red
edge (since the edges of $H$ not in $P'$ are all red). Let $P$ be such a path. Then
$\mc(P) \leq \mc(H) = \mc(P')$. Since $P$
is in $\redpaths{G + I}{F}{u}{v}$, it follows that
$\redbneck{G+I}{F}{u}{v} \leq \mc(P) \leq \mc(P')$, as desired.
\begin{flushright}\qed\end{flushright}\end{proof}

\begin{claim}
\label{claim-eta-geq}
If $\redpaths{G+I}{F}{u}{v}\neq \varnothing$ and
$\redpaths{G' + I'}{F}{u}{v}\neq \varnothing$ then $\redbneck{G+ I}{F}{u}{v} \geq \redbneck{G' + I'}{F}{u}{v}$.
\end{claim}
\begin{proof}
We have to show that $\redbneck{G'+ I'}{F}{u}{v} \leq \mc(P)$
for every path $P \in \redpaths{G+I}{F}{u}{v}$. Consider such a path $P$.
We proceed similarly as in the proof of the previous claim.

If $P$ contains no extra red edge of $G+I$
then $P \in \redpaths{G' + I'}{F}{u}{v}$, and $\redbneck{G' + I'}{F}{u}{v} \leq \mc(P)$
holds. Thus we may assume that $P$ contains at least one such edge.

Let $P'$ be the path obtained from $P$ as follows:
First, for each extra red edge $e^{ab}$ in $P$ with $a,b \neq \etanode{b}$, replace $e^{ab}$ with
$e'^{ab}$. Now, if $P$ includes the vertex $\etanode{b}$, then it has two
extra red edges of the form $e^{a\etanode{b}}$ and $e^{b\etanode{b}}$, respectively,
with $a, b \in \partial(G) \cap \partial(G')$ and $a\neq b$. Replace then the subpath of $P$
consisting of these two edges with the edge $e^{ab}$.
The resulting path $P'$ is in $\redpaths{G'+I'}{F}{u}{v}$. Moreover, it
follows from condition (2) that $\mc(P') \le \mc(P)$. Therefore,
$\redbneck{G'+ I'}{F}{u}{v} \leq \mc(P') \le \mc(P)$, as claimed.
\begin{flushright}\qed\end{flushright}\end{proof}

Lemma~\ref{lem-eta-bneck} follows from Claims~\ref{claim-eta-empty}, ~\ref{claim-eta-leq}, and~\ref{claim-eta-geq}.
\begin{flushright}\qed\end{flushright}\end{proof}

We may now proceed with the proof of Lemma~\ref{lem-opt-eta}.

\begin{proof}[Proof of Lemma~\ref{lem-opt-eta}]
We first show:
\begin{claim}
\label{claim-opt-eta-exists}
There exists a $k$-graph $I'=(K',f',g')$ on the boundary of $G'$
such that $I'$ is $\eta$-compatible with $I$
and $\OPTtw_I(G) \leq \OPTtw_{I'}(G')$.
\end{claim}
\begin{proof}
Let $F\subseteq B(G)$ be a subset of blue edges realizing
$I$ in $G$ such that
$$
\OPTtw_I(G)  = \sum_{uv \in F} \redbneck{G + I}{F}{u}{v}.
$$

Let $I'=(K',f',g')$ be the $k$-graph on the boundary of $G'$ defined
by setting, for every $ab\in E(K')$, $f'(ab) := j$ where $j$ is the index in $[k]$
such that $c_{j} = \intbneck{G'}{F}{a}{b}$, and letting
$g'(ab) := \min \big \{g(ab), \max\{ g'(a\etanode{a}),  g'(\etanode{a}b)\} \big\}$
for every two distinct vertices $a, b\in \partial(G') \setminus \{\etanode{a}\}$, and
$g'(a\etanode{a}) := k$ for every $a \in \partial(G') \setminus \{\etanode{a}\}$.
By definition, the set $F$ realizes $I'$ in $G'$.
Let us show that $I'$ is $\eta$-compatible  with $I$. By definition, $I'$
satisfies conditions (2) and (4) of the definition of $\eta$-compatibility. Also, condition (3)
is satisfied, since $\etanode{b}$ is isolated in $G$.
Thus it remains to show that $f(ab) = \min \{ f'(ab), \max\{f'(a\etanode{a}),f'(\etanode{a}b)\}\}$
for every two distinct vertices $a,b \in \partial(G) \cap \partial(G')$.
Consider two such vertices $a$ and $b$.

First we show that
$f(ab) \leq \min \{ f'(ab), \max\{f'(a\etanode{a}),f'(\etanode{a}b)\}\}$.
If $f'(ab) \leq \max\{f'(a\etanode{a}),f'(\etanode{a}b)\}$, then either $f'(ab) = k$
and the claimed upper bound on $f(ab)$ trivially holds, or $f'(ab) < k$ and hence
there is a path $P' \in \intpaths{G'}{F}{a}{b}$ with $\mc(P') = \intbneck{G'}{F}{a}{b} = c_{f'(ab)}$.
The path $P'$ is also included in $\intpaths{G'}{F}{a}{b}$; hence
$\intbneck{G}{F}{a}{b} \leq c_{f'(ab)}$, which implies $f(ab) \leq f'(ab)$ (since $F$
realizes $I$ in $G$). Now suppose that
$f'(ab) > \max\{f'(a\etanode{a}),f'(\etanode{a}b)\}$. Since the righthand side
of this inequality is strictly less than
$k$, both $\intpaths{G'}{F}{a}{\etanode{a}}$ and $\intpaths{G'}{F}{\etanode{a}}{b}$ are nonempty.
Let $P'_{1} \in \intpaths{G'}{F}{a}{\etanode{a}}$ and $P'_{2} \in \intpaths{G'}{F}{\etanode{a}}{b}$
be paths such that $\mc(P'_{1}) = c_{f'(a\etanode{a})}$ and $\mc(P'_{2}) = c_{f'(\etanode{a}b)}$.
These two paths cannot have any vertex in common other than $\etanode{a}$,
because otherwise their union would contain an $ab$-path $P^{*}$ with $\mc(P^{*}) \leq
\max\{\mc(P'_{1}), \mc(P'_{2})\}$ and avoiding $\etanode{a}$, which would imply
$f'(ab) \leq \max\{f'(a\etanode{a}),f'(\etanode{a}b)\}$, contradicting our hypothesis.
Thus the concatenation of $P'_{1}$ and $P'_{2}$ gives an $ab$-path $P$ which is
internal in $G$ (but not in $G'$) satisfying $\mc(P) = \max\{\mc(P'_{1}), \mc(P'_{2})\}$.
Since $\intbneck{G}{F}{a}{b} \leq \mc(P)$, we deduce that
$f(ab) \leq \max\{f'(a\etanode{a}),f'(\etanode{a}b)\}$, as desired.

Next we prove that
$f(ab) \geq \min \{ f'(ab), \max\{f'(a\etanode{a}),f'(\etanode{a}b)\}\}$.
This is obviously true if $\intpaths{G}{F}{a}{b}$ is empty, so let us assume this is not the case
and let $P\in \intpaths{G}{F}{a}{b}$ be such that $\mc(P) = c_{f(ab)}$.
If $P$ does not include the vertex $\etanode{a}$, then
$P\in \intpaths{G'}{F}{a}{b}$ and hence $\mc(P) \geq \intbneck{G'}{F}{a}{b}$,
implying $f(ab) \geq f'(ab)$.
If $P$ includes $\etanode{a}$, the path $P$ is the concatenation of an $a\etanode{a}$-path $P_{1}$
from $\intpaths{G'}{F}{a}{\etanode{a}}$ with an $\etanode{a}b$-path $P_{2}$ from
$\intpaths{G'}{F}{\etanode{a}}{b}$, implying
$\mc(P) = \max\{\mc(P_{1}), \mc(P_{2})\} \geq \max\{\intbneck{G'}{F}{a}{\etanode{a}},\intbneck{G'}{F}{\etanode{a}}{b}\}$, and hence
$f(ab) \geq \max\{f'(a\etanode{a}),f'(\etanode{a}b)\}$, as desired.

Therefore, $f(ab) = \min \{ f'(ab), \max\{f'(a\etanode{a}),f'(\etanode{a}b)\}\}$
holds, and $I'$ is $\eta$-compatible with $I$.
Now we may apply Lemma~\ref{lem-eta-bneck}, giving
$$
\OPTtw_I(G)= \sum_{uv \in F} \redbneck{G + I}{F}{u}{v}
 = \sum_{uv \in F} \redbneck{G' + I'}{F}{u}{v}
\leq \OPTtw_{I'}(G').
$$
\begin{flushright}\qed\end{flushright}\end{proof}

Next we prove:

\begin{claim}
\label{claim-opt-eta-forall}
$\OPTtw_{I}(G)  \geq \OPTtw_{I'}(G)$
holds for every $k$-graph $I'=(K',f',g')$ on the boundary of $G'$
such that $I'$ is $\eta$-compatible with $I$.
\end{claim}
\begin{proof}
Let $F'\subseteq B(G')$
be a subset of blue edges of $G'$ such that
$$
\OPTtw_{I'}(G')
= \sum_{uv \in F'} \redbneck{G' + I'}{F'}{u}{v}.
$$
By Lemma~\ref{lem-eta-F}, $F:= F'$ realizes
$I$ in $G$. Using again Lemma~\ref{lem-eta-bneck}, we have:
$$
\OPTtw_{I}(G)  \geq \sum_{uv \in F} \redbneck{G + I}{F}{u}{v}
= \sum_{uv \in F} \redbneck{G' + I'}{F}{u}{v}
= \OPTtw_{I'}(G') .
$$
\begin{flushright}\qed\end{flushright}\end{proof}

Lemma~\ref{lem-opt-eta} follows from Claims~\ref{claim-opt-eta-exists} and~\ref{claim-opt-eta-forall}.
\begin{flushright}\qed\end{flushright}\end{proof}

 \subsection{The unary operator $\epsilon$}
 If $G=\epsilon(G')$, then $G$ is obtained from $G'$ by adding an edge $e^{*}$ between
 the two boundary vertices labeled $1$ and $2$.
Notice that $G=G'\oplus H$, where $H$ is the $t$-boundaried graph having only boundary vertices, and only the edge $e^{*}$. Thus, instead of dealing with the $\epsilon$ operator we can use the $\oplus$ operator that we already treated, and introduce two new null-like operators that create the graph
$H$ with the edge $e^{*}$ being either red or blue. Therefore, it is enough to
describe how to compute $\OPTtw_{I}(H)$ for every $k$-graph $I$ on the boundary of $H$
in both cases, which we do now.
    \begin{itemize}
      \item If $e^{*}$ is red with cost $c(e^{*})$ then we have $\OPTtw_{I}(H)=0$ (associated with the acyclic   set $F=\varnothing$ of blue edges) for every $k$-graph $I=(K,f,g)$
      such that $f(e') = c(e^{*})$ and $f(e) = k$ for every $e \in E(K) \setminus \{e'\}$, where
      $e'$ is the edge in $E(K)$ with the same endpoints as $e^{*}$.
      For all other $k$-graphs $I$ we have $\OPTtw_{I}(H)=-\infty$
      (since none of them are realizable in $H$).

      \item If $e^{*}$ is blue then we have $\OPTtw_{I}(H)=0$
      (associated with $F=\varnothing$) for every
      $k$-graph $I=(K,f,g)$ such that $f(e)  = k$ for every $e \in E(K)$.
      In addition, for every
      $k$-graph $I=(K,f,g)$ such that $f(e') = 0$ and
      $f(e)  = k$ for every $e \in E(K) \setminus \{e'\}$ (where $e'$ is defined as previously),
      we have $\OPTtw_{I}(H) = \redbneck{H+I}{F}{a}{b}$ where $F=\{e^{*}\}$ and
      $a, b$ are the two endpoints of $e^{*}$. Let us emphasize that
      the quantity $\redbneck{H+I}{F}{a}{b}$ is
      easily computed here, since it is the minimum of $\mc(P)$ over all
      $ab$-paths $P$ in $H+I$ containing at least one red edge (with $\redbneck{H+I}{F}{a}{b}=c_{k}$
      if there is no such path), and there are at most $t!$ such paths.
      Finally, for all $k$-graphs $I$ not considered above, we have $\OPTtw_{I}(H)=-\infty$.
    \end{itemize}

 \subsection{Unary operators that permute  labels}
 Unary operators that permute the labels of the boundary vertices are handled in the obvious way.

\subsection{The Algorithm}
\label{sec-algo-treewidth}

We may now prove Theorem~\ref{th-tw}, which we restate here.

\setcounter{theoremduplicate}{2} 
\begin{theoremduplicate}
The {\stackmst} problem can be solved in $2^{O(t^3)}m + m^{O(t^2)}$
time on graphs of treewidth $t$.
\end{theoremduplicate}
\begin{proof}
As noted after the definition of $t$-boundaried graphs in the beginning of
Section~\ref{Bounded-Treewidth Graphs}, it is enough to show that
the problem can be solved in $m^{O(t^2)}$ time on a given $t$-boundaried graph
when the construction according to the five operators is also given in input,
thanks to the result of Bodlaender~\cite{Bodlaender1996}.

Our algorithm considers each graph $H$ appearing in the decomposition
tree in a bottom-up fashion, maintaining the $\OPTtw_I(H)$ values (and associated acyclic
sets $F$ of blue edges)
as described by the previous subsections on the five composition operators.

The operators $\oplus$ and $\eta$ require us to check every combination of at most three different
$k$-graphs for compatibility (three for $\oplus$-compatiblity, two for $\eta$-compatibility).
There are $((k+1)^{2})^{t \choose 2}= (k+1)^{t(t-1)}$ different $k$-graphs on a given boundary, so we need to check $O(k^{3t^2})$ combinations. Each check can be done in $O(t^2)$ time.

The most time-consuming check is the one for the $\epsilon$ operator when it adds a blue edge, since
the computation of $\OPTtw_I(H)$ for one $k$-graph $I$
may require considering $O(t!)$ paths.

The total time complexity of the algorithm is therefore bounded by
$O(k^{3 t^2}\cdot t!)=m^{O(t^2)}$.

This results in a polynomial-time algorithm, when the input graph is of bounded treewidth, for computing the maximum revenue
achievable by the leader. Moreover, as mentioned earlier, it is not difficult to keep track of a witness $F \subseteq B(H)$ for $\OPTtw_I(H)$ whenever $\OPTtw_I(H) > -\infty$ when applying any one of the five operators.
\begin{flushright}\qed\end{flushright}\end{proof}

\section{Conclusion and Open Problems}

To our knowledge, our algorithms are the first examples of a bilevel pricing
problem solved by dynamic programming on a graph decomposition tree. Several interesting problems are left open.

We proved that the problem can be solved in polynomial time for every constant value of the treewidth $t$. However,
it is unclear whether there exists a fixed-parameter algorithm of complexity $O(f(t) n^c)$ for an arbitrary (possibly large)
function $f$ of $t$ and a constant $c$. In fact, we conjecture that under reasonable complexity-theoretic assumptions,
such an algorithm does not exist.

We believe that our results provide insights into the structure of the problem,
and could be a stepping stone toward a polynomial-time approximation scheme
for planar graphs. Also, the proposed techniques could be useful in the design of
dynamic programming algorithms for other important pricing problems in graphs, including
pricing problems with many followers~\cite{BHK08,GLSU09-journal}, and
Stackelberg problems involving shortest paths~\cite{RSM05,BCKLN10}
or shortest path trees~\cite{BGPW08}.

\paragraph{Acknowledgements.} 
We would like to thank the anonymous referees for their helpful comments. 

\bibliographystyle{alpha}
\bibliography{planarStackMST}

\end{document}